\documentclass[12pt]{article}
\usepackage{amssymb,amsfonts,latexsym,amsmath,theorem,mathrsfs,graphicx,xcolor}
\usepackage[square,sort&compress,numbers]{natbib}
\usepackage{subcaption}
\usepackage{hyperref}
\usepackage{float, braket, mathrsfs}
\usepackage{bbold}
\usepackage{bm}
\usepackage{tikz}
\usepackage{tikz-3dplot}
\usepackage{tikz-cd}
\usepackage[utf8]{inputenc}
\usepackage[english]{babel}
\usepackage{imakeidx}

\DeclareMathOperator{\grad}{grad}
\DeclareMathOperator{\Aut}{Aut}

\DeclareMathOperator{\Tr}{Tr}

\DeclareMathOperator{\diag}{diag}
\DeclareMathOperator{\musicd}{\flat}

\newcommand{\Sym}{{\rm Sym}}

\newcommand{\OO}{\mathscr{O}} 
\newcommand{\MM}{\mathscr{M}}

\newcommand{\evec}{{\bf e}}

\newcommand{\xvec}{{\bf x}} 
\newcommand{\vvec}{{\bf v}}

\newcommand{\g}{\mathfrak{g}}

\newcommand{\R}{{\mathbb{R}}}

\newcommand{\Z}{{\mathbb{Z}}}

\newcommand{\I}{{\mathbb{I}}}
\newcommand{\T}{{\mathbb{T}}}

\newcommand{\beq}{\begin{equation}}
\newcommand{\eeq}{\end{equation}}
\newcommand{\bea}{\begin{eqnarray}}
\newcommand{\eea}{\end{eqnarray}}
\newcommand{\ben}{\begin{eqnarray*}}
\newcommand{\een}{\end{eqnarray*}}
\newcommand{\bem}{\begin{enumerate}}
\newcommand{\eem}{\end{enumerate}}

\newcommand{\ra}{\rightarrow}

\newcommand{\cd}{\partial}
\newcommand{\wt}{\widetilde}

\newcommand{\less}{\backslash}

\newcommand{\su}{{\mathfrak{su}}}
\newcommand{\SU}{{\mathrm{SU}}}

\newcommand{\SO}{{\mathrm{SO}}}

\newcommand{\SPD}{{\mathrm{SPD}}}

\def \d{\mathrm{d}}

\newcommand{\ip}[1]{\langle #1 \rangle}
\newcommand{\ignore}[1]{}

\newcommand{\vol}{{\rm vol}}

\newcommand{\lamvec}{{\mbox{\boldmath{$\lambda$}}}}

\newcommand{\nvec}{\mbox{\boldmath{$n$}}}

\newcommand{\Xvec}{\mbox{\boldmath{$X$}}}
\newcommand{\Lvec}{\mbox{\boldmath{$L$}}}

\newcommand{\eps}{\varepsilon}
\renewcommand{\phi}{\varphi}

\theoremstyle{plain}
\newtheorem{thm}{Theorem}

\newtheorem{prop}[thm]{Proposition}

\newtheorem{cor}[thm]{Corollary}

{\theorembodyfont{\rmfamily}
\newtheorem{defn}[thm]{Definition}

}

\newenvironment{proof}{\noindent{\it Proof:\, }}{\hfill$\Box$\vspace*{0.5cm}
}

\begin{document}

\title{Skyrme crystals with massive pions}

\author{ Derek Harland\footnote{d.g.harland@leeds.ac.uk},
Paul Leask\footnote{mmpnl@leeds.ac.uk  (corresponding  author)},\: and Martin Speight\footnote{j.m.speight@leeds.ac.uk} \\
School of Mathematics, University of Leeds, Leeds LS2 9JT, UK}
\date{16th May 2023}

\maketitle

\begin{abstract}
The crystalline structure of nuclear matter is investigated in the standard Skyrme model with massive pions.
A semi-analytic method is developed to determine local minima of the static energy functional with respect to variations of both the field and the period lattice of the crystal.
Four distinct Skyrme crystals are found. Two of these were already known -- the cubic lattice of half-skyrmions and the $\alpha$-particle crystal -- but two are new. These new solutions have lower energy per baryon number and less symmetry, being periodic with respect to trigonal but {\em not} cubic period lattices.
Minimal energy crystals are also constructed under the constraint of constant baryon density, and its shown that the two new non-cubic crystals tend to chain and multi-wall solutions at low densities.
\end{abstract}


\section{Introduction}

The Skyrme model \cite{Skyrme_1961} is a nonlinear field theory of $\pi$ mesons that accommodates nucleons as topological solitons. It emerges as a low energy effective theory of QCD in the regime where the number of colours (or equivalently, the rank of the gauge group) grows large \cite{Witten_1983_1,Witten_1983_2}.
The model has only one field, taking values in the Lie group $\SU(2)$. Field configurations are classified topologically by an integer-valued homotopy
invariant $B$ which is interpreted physically as the baryon number of the configuration. There is a topological lower bound on static energy of the form
$E\geq E_\textrm{top}B$, where $E_\textrm{top}$ is some positive constant, originally due to
Faddeev \cite{fad-bound} and subsequently improved by one of us \cite{har1}. (Improved in this context means that the constant $E_\textrm{top}$ is increased.) Let $E(B)$ denote the minimum static energy among all fields of baryon number $B$. The
energy bound $E(B)=E_\textrm{top}B$ is never attained, but numerical studies suggest that the ratio
$E(B)/B$ decreases monotonically, and hence converges to some limit $E_*$ as $B\ra\infty$. This suggests that, as $B$ grows large, minimal energy Skyrme fields may tend to some regular, spatially periodic crystalline structure, with baryon number $B$ and energy $E_*B$ per unit cell.

Such a crystal structure was first proposed by Klebanov \cite{Klebanov_1985}.
He found a crystal of $B=1$ skyrmions arranged in a simple cubic (SC) lattice so that every unit skyrmion is internally oriented to be in the attractive channel with respect to its nearest neighbours.
Manton snd Goldhaber \cite{Goldhaber_1987} later found that at high densities this crystal undergoes a phase transition to a body centred cubic (BCC) lattice of half-skyrmions with a lower energy per baryon ($E/B$) than Klebanov's SC crystal.
Then, independently, Kugler and Shtrikman \cite{Kugler_1988} and Castillejo \textit{et al.} \cite{Castillejo_1989} determined a new  solution with 
lower $E/B$, wherein skyrmions are initially arranged in a face centred cubic (FCC) lattice and relax to a SC lattice of half-skyrmions.
In all of these studies, the model has massless pions, $m_\pi=0$, the unit cell has baryon number $B=4$, and the energy functional has been varied only over
cubic period lattices, that is, only the side length of the cube is varied.

The phase structure of the massless pion Skyrme model has been studied by Jackson and Verbaarschot \cite{Jackson_1988}. Perapechka and Shnir \cite{Shnir_2017} investigated phase transitions in the Skyrme model with $m_\pi>0$. (They also considered the effect of incorporating a non-standard sextic term in the energy functional. Such terms are of current interest because they arise in so-called ``near BPS" variants of the model \cite{adasanwer2},
but will not be of central relevance to our considerations.)
Two candidates have been previously proposed as the minimal $E/B$ crystal with massive pions: the cubic lattice of half-skyrmions \cite{Kugler_1988,Castillejo_1989} and the $\alpha$-particle lattice \cite{Feist_2012}. Once again, these studies impose periodicity with respect to a cubic period lattice, and vary only the side length of the cube.

In this paper, we study Skyrme crystals in the model with $m_\pi>0$, minimizing the energy with respect to variations of \textit{both} the Skyrme field $\phi:\R^3/\Lambda\ra \SU(2)$ \textit{and} its period lattice $\Lambda$. To achieve this, we identify every 3-torus $\R^3/\Lambda$ with the fixed
3-torus $\T^3=S^1\times S^1\times S^1=\R^3/\Z^3$ by means of the obvious
diffeomorphism $f:\T^3\ra \R^3/\Lambda$, and equip $\T^3$ with the pullback of the Euclidean metric on $\R^3/\Lambda$, $g=f^*g_\textrm{Euc}$. Varying over all
period lattices $\Lambda$ is then equivalent to varying over all flat metrics
$g$ on the fixed torus $\T^3$. This approach was introduced in
\cite{Speight_2014}, in which the interpretation of the gradient of the energy with respect to the metric $g$ as the stress tensor of the field
was repeatedly exploited. In the current paper, we will find it convenient to think of the metric variational problem more concretely, by identifying $g$ with the constant
symmetric positive definite matrix $(g_{ij})$ representing it with respect to the canonical coordinate system on $\T^3$. 

So, the numerical task we set ourselves is, for fixed topological degree $B$, to minimize $E(\phi,g)$ among all degree $B$ maps $\phi:\T^3\ra \SU(2)$ and flat metrics $g$ on $\T^3$. It is known that, for fixed $g$, the function
$\phi\mapsto E(\phi,g)$ attains a minimum (in a function space of rather low
regularity) \cite{auckap}. The complementary problem of minimising in $g$ for fixed $\phi$ was first studied in \cite{Speight_2014} and numerically implemented in \cite{Leask_2022, babspewin}.
In \cite{Speight_2014} it was shown that in a two-dimensional toy model (the baby Skyrme model), any critical metric $g$ is automatically a local minimimum of $E$.  The problem of extending this result to the Skyrme model was discussed, but unfortunately the proof used in two dimensions did not generalise.  Existence of critical metrics was not addressed in \cite{Speight_2014}.

In the present paper we obtain a much stronger result.  We show in Corollary \ref{cor4} that, for fixed $\phi$ satisfying very mild assumptions, there is a unique
flat metric with respect to which $E(\phi,g)$ is minimal, and hence a unique period lattice $\Lambda$ (up to automorphism) with respect to which
$\phi$ has minimal energy per unit cell. In the special case $m_\pi=0$, we can even write down this metric explicitly. In the
(more interesting) massive case, we can resort to a gradient based
numerical minimization scheme to find $g$. Applying a similar scheme to minimize over $\phi:\T^3\ra\SU(2)$ in tandem, we can find the energetically optimal field and period lattice for a given $B$, without ever imposing any symmetry assumptions on the lattice. 


The results reveal that, for $m_\pi>0$, the energetically optimal lattice (with $B=4$ per unit cell) does {\em not} have cubic symmetry. In fact there are two
crystal solutions with trigonal period lattices (orthorhombic with
side lengths $L_1=L_2\neq L_3$) which have lower energy than the lowest
strictly cubic lattice. This fact persists if, instead of minimizing over all flat $g$, we minimize only over the subset of metrics with fixed total volume. This is equivalent to minimizing under the constraint of fixed average baryon density, a problem of phenomenological interest \cite{Adam_2022}. As might be expected, the energy difference between the trigonal and cubic lattices becomes negligible
as baryon density grows very large, but is significant at lower densities.

The rest of the paper is structured as follows. In section~\ref{sec: Skyrme model}, we formulate the model mathematically, considering in detail how its energy functional depends on the metric on physical space.
In section~\ref{sec:eumm}, we prove existence and uniqueness of an energy minimizing metric $g$ for any given fixed field. In section~\ref{sec:anf} we describe our numerical scheme in detail, while section~\ref{sec: Skyrme crystal solutions} presents the results of this scheme. In section~\ref{sec: Fixed baryon density variations} we determine minimal energy crystals under the constraint of fixed baryon density. Finally, section~\ref{sec:conc} presents some concluding remarks.


\section{The Skyrme model}
\label{sec: Skyrme model}

We wish to study the Skyrme model under the assumption that the Skyrme field
$\phi:\R^3\ra\SU(2)$ is periodic with respect to some $3$-dimensional lattice
\beq
\Lambda=\{n_1\Xvec_1+n_2\Xvec_2+n_3\Xvec_3:\nvec\in\Z^3\},
\eeq
 that is,
$\phi(\xvec+\Xvec)=\phi(\xvec)$ for all $\xvec\in\R^3$ and $\Xvec\in\Lambda$, where $\Xvec_1,\Xvec_2,\Xvec_3$ is an oriented basis for $\R^3$,. In this case, we may equally well interpret the field as a map $\phi:\R^3/\Lambda\ra\SU(2)$, where the torus $\R^3/\Lambda$ inherits a metric $g_\Lambda$ from the Euclidean metric on $\R^3$. It is convenient to identify $\R^3/\Lambda$ with
the standard torus $\T^3=\R^3/\Z^3$ via the diffeomorphism
\beq
f_\Lambda:\T^3\ra \R^3/\Lambda,\qquad (x_1,x_2,x_3)\mapsto x_1\Xvec_1+x_2\Xvec_2
+x_3\Xvec_3,
\eeq
and denote by $g$ the pullback of the metric $g_\Lambda$ by $f_\Lambda$. Explicitly,
\beq\label{sg1}
g=g_{ij}\d x_i \d x_j,\qquad g_{ij}=\Xvec_i\cdot\Xvec_j.
\eeq
Note that the matrix $(g_{ij})$ is a symmetric, positive definite real $3\times 3$ matrix. We denote the set of such matrices $\SPD_3$ and note that every such matrix arises as the metric on $\T^3$ corresponding to some lattice $\Lambda$ and lattices producing the same matrix are related by an oriented isometry of $\R^3$. Hence, instead of considering the Skyrme model on $\R^3/\Lambda$ for all lattices $\Lambda$, we may restrict to the standard torus $\T^3$ but consider all flat metrics on $\T^3$, $g=g_{ij}\d x_i\d x_j$, $(g_{ij})\in \SPD_3$. 
It is convenient to abuse notation slightly and denote the matrix $(g_{ij})$ by $g$.   

The energy of a Skyrme field $\phi:(\T^3,g)\ra \SU(2)$ is
defined in stages as follows. Denote by $\mu$ the left Maurer-Cartan form on $\SU(2)$, that is, the $\su(2)$-valued one-form that maps a tangent vector $X\in T_U\SU(2)$ to the vector $\mu(X)\in T_{\I_2}\SU(2)$ whose image under left translation by $U$ is $X$. The pullback of $\mu$ by $\phi$,
\beq
\phi^*\mu=\phi^{-1}\d\phi=: L_i\d x_i,
\eeq
is usually called the {\em left current}. We equip $\su(2)$ with the usual $Ad$ invariant inner product $(X,Y)_{\su(2)}=-\frac12\Tr(XY)$ and define the Dirichlet energy of $\phi$ to be
\beq
E_2(\phi,g)=\int_{\T^3}|\phi^*\mu|_g^2\vol_g
=-\frac12\int_{\T^3}g^{ij}\Tr(L_iL_j)\sqrt{|g|}\vol_0
\eeq
where $g^{ij}$ are the components of the inverse matrix $g^{-1}$, 
$|g|=\det g$ and $\vol_0=\d x_1\wedge \d x_2\wedge \d x_3$ is the canonical volume form on $\T^3$. We are following the standard convention that a numerical subscript on an energy functional denotes the degree of its integrand considered as a polynomial in spatial derivatives. It is important to note that we regard $E_2$ as a function of both $\phi$ and $g$. 

To define the Skyrme term $E_4(\phi,g)$, we introduce the $\su(2)$-valued two-form
$\Omega$ on $\SU(2)$ by $\Omega(X,Y)=[\mu(X),\mu(Y)]$. Then
\beq
E_4(\phi,g)=\frac14\int_{\T^3}|\phi^*\Omega|_g^2\vol_g
=-\frac{1}{16}\int_{\T^3}g^{ik}g^{jl}\Tr([L_i,L_j][L_k,L_l])\sqrt{|g|}\vol_0.
\eeq
We will also include a potential term
\beq
E_0(\phi,g)=\int_{\T^3}V(\phi)\vol_g=\int_{\T^3}V(\phi)\sqrt{|g|}\vol_0
\eeq
where $V:\SU(2)\ra [0,\infty)$ is some smooth function. The usual choice is
\beq
V(U)=m_\pi^2\Tr(\I_2-U),
\eeq 
which has the effect of giving the pions of the theory (small amplitude waves about the vacuum $\phi=\I_2$) mass $m_\pi$. 

In summary, the Skyrme energy of a field $\phi:\T^3\ra SU(2)$ and metric
$g\in \SPD_3$ is
\beq
E(\phi,g)=E_2(\phi,g)+E_4(\phi,g)+E_0(\phi,g).
\eeq
Since $\SU(2)$ is diffeomorphic to $S^3$ the homotopy class of the field $\phi$ is, by the Hopf degree theorem, labelled by its topological degree $B\in\Z$,
\beq
B=\frac{1}{2\pi^2}\int_{\T^3}\phi^*\vol_{\SU(2)},
\eeq
where $\vol_{\SU(2)}$ is the volume form on $\SU(2)$ defined by the bi-invariant metric $h$ which, at $U=\I_2$ coincides with $(\cdot,\cdot)_{\su(2)}$ (equivalently, the round metric of radius $1$ on $S^3$). The mathematical problem 
that this paper addresses is to minimize $E(\phi,g)$, with respect to both 
$\phi$ and $g$, among all fields of fixed degree $B$ and all metrics $g\in \SPD_3$.


\section{Existence and uniqueness of minimizing metrics}
\label{sec:eumm}

For fixed $g\in\SPD_3$, $(\T^3,g)$ is a fixed, compact oriented Riemannian $3$-manifold, and it follows from a direct application of the calculus of variations that the
functional $\phi\mapsto E(\phi,g)$ attains a minimum in each degree class in the space of
finite energy maps in the Sobolev space $W^{1,2}(\T^3,\SU(2))$ \cite{auckap}. 
In this section we address the complementary variational problem: we fix a map
$\phi:\T^3\ra\SU(2)$ and establish existence, and uniqueness, of a minimizer of the function $\SPD_3\ra\R$, $g\mapsto E(\phi,g)$ which, for brevity, we will denote $E(g)$. 

It is convenient (and makes the result more comprehensive) to include the possibility of a sextic term in the Skyrme energy,
\beq
E_6(\phi,g)=\int_{\T^3}|\phi^*\Xi|^2\vol_g
\eeq
where $\Xi$ is some $3$-form on $\SU(2)$, for example, a constant multiple of 
$\vol_{\SU(2)}$. Terms of this kind are of phenomenological interest since they arise in so-called near BPS variants of the Skyrme model \cite{adasanwer2}. To specialize to the model of primary interest, we simply choose $\Xi=0$. 

We begin by analyzing in more detail the $g$ dependence of the terms in $E$. We first note that
\beq
    E_2(\varphi,g)=g^{ij}\sqrt{\det g}\int_{\T^3}(L_i,L_j)_{\su(2)}\vol_0=\sqrt{\det g}\Tr(Hg^{-1})
\eeq
where
\beq
H_{ij}(\varphi)=\int_{\T^3}(L_i,L_j)_{\su(2)}\vol_0
\eeq
is a symmetric positive semi-definite matrix depending on $\phi$ but independent of $g$. Let us assume that $\phi$ is $C^1$ (so that this matrix is well-defined) and is immersive somewhere, meaning that there is some point $p\in \T^3$ at which $\d\phi_p$ is invertible. Note that this follows immediately for all maps with $B\neq 0$.
By continuous differentiability, it follows that $\phi$ is immersive on some neighbourhood of $p$. Then the matrix $H$ is actually positive definite, for if not, there exists
$\vvec\in\R^3$ with $\vvec\cdot H\vvec=0$, whence
\beq
\int_{\T^3}|\vvec\cdot\Lvec|_{\su(2)}^2\vol_0=0
\eeq
and hence $\d\phi(v_i\cd/\cd x_i)=0$ almost everywhere. This contradicts immersivity of $\phi$ on a neighbourhood of $p$. We conclude that $H\in \SPD_3$.  

To understand $E_4(g)$ we appeal to an isomorphism peculiar to $3$ dimensions. Having chosen a $3$-form $\vol_0=\d x_1\wedge\d x_2 \wedge \d x_3$ on $\T^3$, there is an isomorphism $(T_p\T^3)\otimes\su(2)\ra (\Lambda^2 T^*_p \T^3)\otimes \su(2)$ defined by $X\mapsto \iota_X\vol_0$. Denote by $X_\phi$ the section of $T\T^3\otimes \su(2)$ whose image under this isomorphism is $\phi^*\Omega$, 
\beq
\phi^*\Omega=\iota_{X_\phi}\vol_0,
\eeq
and note that $X_\phi$ depends on
$\phi$, but is independent of $g$. We may similarly define a $g$-dependent section $X_\phi^g$ of $T\T^3\otimes\su(2)$ by using the isomorphism between $2$-forms and tangent vectors defined by $\vol_g$ instead of $\vol_0$:
\beq
\phi^*\Omega=\iota_{X_\phi^g}\vol_g=\iota_{\sqrt{|g|}X_\phi^g}\vol_0.
\eeq
Clearly $X_\phi^g=X_\phi/\sqrt{|g|}$. An alternative interpretation of $X_\phi^g$ is that it is the vector field metrically dual to the Hodge dual of $\phi^*\Omega$ with respect to the metric $g$, that is,
\beq
*_g \phi^*\Omega=\frac{1}{\sqrt{|g|}}\musicd_g X_\phi,
\eeq
where $*_g:\Lambda^2T^*\T^3\ra T^*\T^3$ denotes the Hodge isomorphism  and $\musicd_g:T\T^3\ra T^*\T^3$ the metric isomorphism defined by $g$. Hence
\beq
E_4(\varphi,g)=\frac14\int_{\T^3}|X_\phi/\sqrt{|g|}|_g^2\vol_g
=\frac{1}{4\sqrt{|g|}}\int_{\T^3}g(X_\phi,X_\phi)\vol_0
=\frac{1}{\sqrt{\det g}}\Tr(Fg)
\eeq
where
\beq
F_{ij}(\varphi)=\frac14\int_{\T^3}(X_{\phi,i},X_{\phi,j})_{\su(2)}\vol_0
\eeq
is another symmetric positive semi-definite matrix depending on $\phi$ but independent of $g$, and $X_\phi=:X_{\phi,i}\cd/\cd x_i$. In terms of the left currents
\beq
X_{\phi,i}=\frac12\epsilon_{ikl}[L_k,L_l],
\eeq
and so the matrix $F$ takes the explicit form
\beq
F_{ij}(\varphi)=-\frac{1}{32}\epsilon_{ikl}\epsilon_{jmn}\int_{\T^3}
\Tr([L_k,L_l][L_m,L_n])\vol_0.
\eeq
Once again, our non-degeneracy assumption on $\phi$ (that it is $C^1$ and somewhere immersive) implies that $F$ is positive definite. For if not, there exists $\vvec\in\R^3$ such that
$\vvec\cdot F\vvec=0$, whence $\vvec\cdot X_\phi=0$ and so
$*_g\phi^*\Omega(v_i\cd/\cd x_i)=0$. But then $\phi^*\Omega$
vanishes on every plane in $T\T^3$ $g$-orthogonal to $\vvec$, which contradicts nondegeneracy of $\Omega$ and immersivity of $\phi$.

The remaining terms of $E$ are more straightforward.
\beq
E_0(\varphi,g)=\int_{\T^3}V(\phi)\vol_g=C_0\sqrt{\det g},
\eeq
where
\beq
C_0(\varphi)=\int_{\T^3}V(\phi)\vol_0\geq 0
\eeq
is a constant. Finally, we note that $\phi^*\Xi=f_\Xi\vol_0$ for some real function $f_\Xi:\T^3\ra\R$ independent of $g$. Then 
\beq
f_\Xi\vol_0=(*_g\phi^*\Xi)\vol_g=(*_g\phi^*\Xi)\sqrt{\det g}\vol_0,
\eeq
so $*_g\phi^*\Xi=f_\Xi/\sqrt{\det g}$. Hence
\beq
E_6(\varphi,g)=\int_{\T^3}\phi^*\Xi\wedge *_g\phi^*\Xi=\frac{C_6}{\sqrt{\det g}}
\eeq
where
\beq
C_6(\varphi)=\int_{\T^3}f_\Xi^2\vol_0\geq 0
\eeq
is a constant. Note that we allow the possibility that $C_0$ or $C_6$ is $0$, to accommodate versions of the model with no potential or sextic term. 

In summary, for a fixed $C^1$ map $\phi:\T^3\ra\SU(2)$ which is immersive somewhere, the total Skyrme energy as a function of the metric $g$ on $\T^3$ is
\beq\label{scgi}
E(g):=E(\left.\varphi\right|_\textrm{fixed},g)=\sqrt{\det g}\Tr(Hg^{-1})+\frac{1}{\sqrt{\det g}}\Tr(Fg)
+C_0\sqrt{\det g}+\frac{C_6}{\sqrt{\det g}},
\eeq
where $H,F\in\SPD_3$ and $C_0,C_6\in[0,\infty)$ are constants. 
We wish to prove that the function $E:\SPD_3\ra\R$ attains a unique global minimum, and has no other critical points. Before doing so, we note that 
$E=\wt{E}\circ \sigma$ where
\beq\label{scgii}
\wt{E}:\SPD_3\ra\R,\qquad \wt{E}(\Sigma)=\Tr(H\Sigma^{-1})+\Tr(F\Sigma)
+\frac{C_0}{\det\Sigma}+{C_6}{{\det \Sigma}}
\eeq
and $\sigma$ is the map
\beq
\sigma:\SPD_3\ra\SPD_3,\qquad g\mapsto\Sigma=\frac{g}{\sqrt{\det g}}.
\eeq
Since $\sigma$ is a diffeomorphism, we may equivalently prove that $\wt{E}:\SPD_3\ra\R$ attains a unique global minimum and has no other critical points. We do this in two stages.

\begin{prop} \label{prop1}
The function $\wt{E}:\SPD_3\ra\R$ of equation \eqref{scgii} attains a global
minimum.
\end{prop}

\begin{proof} Clearly $\wt{E}$ is bounded below (by $0$). Let $E_*=\inf \wt{E}$. We must show that there exists $\Sigma_*\in\SPD_3$ with $\wt{E}(\Sigma_*)=E_*$. 

 Consider the map
\beq
f:(0,\infty)^3\times O(3)\ra\SPD_3,\qquad
f(\lamvec,\OO)=\OO D_\lamvec\OO^T,\quad 
D_\lamvec:=\diag(\lambda_1,\lambda_2,\lambda_3).
\eeq
This map is surjective: given any $\Sigma\in\SPD_3$ we may take $\lambda_i$ to be its eigenvalues and $\OO$ to be an orthogonal matrix whose columns are its corresponding eigenvectors. Hence, it suffices to prove that 
\beq
(\wt{E}\circ f)(\lamvec,\OO)=
\Tr(\OO^TH\OO D_\lamvec^{-1})+\Tr(\OO^TF\OO D_\lamvec)
+\frac{C_0}{\lambda_1\lambda_2\lambda_3}+C_6\lambda_1\lambda_2\lambda_3
\eeq
attains the value $E_*$.

Let $(\lamvec_n,\OO_n)$ be a sequence in $(0,\infty)^3\times O(3)$ such that
\beq
(\wt{E}\circ f)(\lamvec_n,\OO_n)\ra E_*.
\eeq
Such a sequence exists since $f$ is surjective. Consider the six functions
\beq
(\OO^TH\OO)_{ii}, (\OO^TF\OO)_{ii}:O(3)\ra (0,\infty),
\qquad i=1,2,3,
\eeq
mapping $\OO$ to the diagonal entries of $\OO^TH\OO,
\OO^TF\OO\in\SPD_3$, and note that these functions are strictly positive since $H,F$ are positive definite. Since these functions are smooth and $O(3)$ is compact, there exists $\alpha>0$ such that, for all $\OO\in O(3)$, $(\OO^TH\OO)_{ii}, (\OO^TF\OO)_{ii}\geq \alpha$. Hence, for all $(\lamvec,\OO)\in(0,\infty)^3\times O(3)$,
\beq\label{thebound}
(\wt{E}\circ f)(\lamvec,\OO)
\geq \alpha\left(\frac{1}{\lambda_1}+\frac{1}{\lambda_2}+\frac{1}{\lambda_3}+\lambda_1+\lambda_2+\lambda_3\right).
\eeq
We may assume that $(\wt{E}\circ f)(\lamvec_n,\OO_n)\leq E_*+1$ for all $n$, so $\lamvec_n\in[K^{-1},K]^3$ for all $n$, where $K=E_*/\alpha$. Hence, the sequence $(\lamvec_n,\OO_n)$ takes values in a compact subset of $(0,\infty)^3\times O(3)$, and so has a convergent subsequence, converging to $(\lamvec_*,\OO_*)$ say, which, without loss of generality, we may assume is $(\lamvec_n,\OO_n)$ itself. So
$(\wt{E}\circ f)(\lamvec_n,\OO_n)\ra E_*$ and
$(\lamvec_n,\OO_n)\ra (\lamvec_*,\OO_*)$. But $(\wt{E}\circ f)$ is continuous, so $(\wt{E}\circ f)(\lambda_*,\OO_*)=E_*$.

It follows that $\wt{E}(\Sigma_*)=E_*$ where
\beq
\Sigma_*=\OO_*D_{\lamvec_*}\OO^T_*,
\eeq
which completes the proof.
\end{proof}

We note in passing that the minimizing {\em metric} whose existence follows from Proposition \ref{prop1} is
\beq
g_*=\sigma^{-1}(\Sigma_*)=\frac{\Sigma_*}{\det\Sigma_*}.
\eeq

It remains to prove that $\wt{E}$ has no other critical points. We achieve this by proving that $\wt{E}$ is {\em strictly
convex}, in the following sense:

\begin{defn}\label{scdef}
A function $f:M\ra\R$ on a Riemannian manifold $M$ is 
{\bf convex} if, for all non-constant geodesics $\gamma(t)$ in $M$,
$(f\circ\gamma)''(t)\geq0$, and {\bf strictly convex} if, for all such geodesics, $(f\circ\gamma)''(t)>0$.
\end{defn}

To apply this definition to $\wt{E}$, we must equip $\SPD_3$ with a Riemannian metric, $G$. The correct choice for our purposes is 
\beq
G_\Sigma:T_\Sigma\SPD_3\times T_\Sigma\SPD_3\ra\R,\qquad
G_\Sigma(\xi_1,\xi_2)=\Tr(\Sigma^{-1}\xi_1\Sigma^{-1}\xi_2),
\eeq
where we have identified $T_\Sigma\SPD_3$ with $\Sym_3$, the vector space of symmetric $3\times 3$ real matrices. We will exploit several useful properties of the metric $G$, established in
\cite{Pennec_2006}.
It is invariant under the $GL(3,\R)$ action
\beq
GL(3,\R)\times\SPD_3\ra\SPD_3,\qquad 
(A,\Sigma)\mapsto A\Sigma A^T
\eeq
on $\SPD_3$. It is also inversion invariant, that is,
\beq
\iota:\SPD_3\ra\SPD_3,\qquad \iota(\Sigma)=\Sigma^{-1}
\eeq
is an isometry. The general geodesic through $\I_3$ takes the form
\beq
\gamma(t)=\exp(t\xi),\qquad \xi\in\Sym_3,
\eeq
and hence a general nonconstant geodesic through $\Sigma$ is
\beq\label{ag}
\gamma(t)=A\exp(t\xi)A^T,
\eeq
where $A\in GL(3,\R)$ satisfies $AA^T=\Sigma$, and $\xi\neq 0$. 
Finally, it is complete and between any pair of distinct points $\Sigma_1$, $\Sigma_2$, there is a geodesic, unique up to parametrization.

\begin{prop}\label{prop2} 
The function $\wt{E}:\SPD_3\ra\R$ of equation \eqref{scgii} is strictly convex with respect to the metric $G$.
\end{prop}

\begin{proof} Given a constant $M\in\SPD_3$, consider the function 
\beq
f_M:\SPD_3\ra\R,\qquad f_M(\Sigma)=\Tr(M\Sigma).
\eeq
 Let $\gamma$ be an arbitrary non-constant geodesic, as in \eqref{ag}. Then
\bea
(f_M\circ \gamma)''(0)&=&\frac{d^2\:}{dt^2}\bigg|_{t=0}
\Tr(MA\exp(t\xi)A^T)
=\Tr(MA\xi^2A^T)\nonumber \\
&=&\Tr((A\xi)^TM(A\xi))
=\sum_{i=1}^3 \vvec_i\cdot M\vvec_i
\eea
where $\vvec_i$ are the columns of $A\xi$. Since $M$ is positive definite, it follows that $(f_M\circ \gamma)''(0)\geq 0$, and equals $0$ only if $\vvec_1=\vvec_2=\vvec_3=0$. But
$\xi\neq 0$ (the geodesic is nonconstant) so at least one $\vvec_i\neq 0$. Hence $(f_M\circ\gamma)''(0)>0$ for all
nonconstant geodesics $\gamma$. It follows that $(f_M\circ\gamma)''(T)>0$ for all nonconstant geodesics and all $T\in\R$, since for all geodesics $\gamma$ and constants $T$, $\wt\gamma(t)=\gamma(t+T)$ is a geodesic. 

Similarly, $\det:\SPD_3\ra\R$ is convex since, for all nonconstant geodesics
\bea
(\det\circ\gamma)''(0)&=&
\frac{d^2\:}{dt^2}\bigg|_{t=0}(\det A)^2\det\exp(t\xi)
=(\det A)^2\frac{d^2\:}{dt^2}\bigg|_{t=0}\exp(t\Tr\xi)\nonumber \\
&=&(\det A)^2(\Tr\xi)^2\geq 0
\eea

It follows that 
\beq
\wt{E}=f_H\circ\iota+f_F+C_0\det\circ\iota+ C_6\det
\eeq
is strictly convex, since $H,F\in\SPD_3$, $\iota$ is an isometry, and $C_0,C_6\geq 0$.
\end{proof}

Propositions \ref{prop1} and \ref{prop2} quickly yield the desired result.

\begin{cor}\label{cor4} Let $\phi:\T^3\ra\SU(2)$ be a fixed $C^1$ map that is immersive somewhere. Then the function $\SPD_3\ra\R$ mapping a flat metric $g$ on $\T^3$ to the Skyrme energy $E(\phi,g)$ attains a unique global minimum, and has no other critical points.
\end{cor}

\begin{proof} As previously established $E(\phi,g)=\wt{E}(\sigma(g))$ where $\wt{E}$ is the function defined in \eqref{scgii} and $\sigma$ is a diffeomorphism of $\SPD_3$. 
By Proposition \ref{prop1},
$\wt{E}$ attains a minimum at some $\Sigma_*\in\SPD_3$, whence
$E$ attains a global minimum at $g_*=\sigma^{-1}(\Sigma_*)$. Assume, towards a contradiction, that $E$ has a second critical point $g_{**}\neq g_*$. Then $\wt{E}$ has a second critical point
at $\Sigma_{**}=\sigma(g_{**})\neq\Sigma_*$. Let $\gamma:[0,1]\ra\SPD_3$ be a geodesic (with respect to $G$) with $\gamma(0)=\Sigma_*$ and $\gamma(1)=\Sigma_{**}$. Then, by Rolle's Theorem applied to $(\wt{E}\circ\gamma)':[0,1]\ra\R$, there exists $t\in(0,1)$ at which $(\wt{E}\circ\gamma)''(t)=0$. But this contradicts Proposition \ref{prop2}. 
\end{proof}

In the case of the standard Skyrme model without a potential ($V=0$ and $\Xi=0$), we can find the minimizing metric $g_*$ explicitly. We note that, in this case
\beq
\wt{E}(\Sigma)=\Tr(H\Sigma^{-1}+F\Sigma),
\eeq
whence
\beq
\d\wt{E}_\Sigma(\xi)=\Tr(-H\Sigma^{-1}\xi\Sigma^{-1}+F\xi)
=\Tr((F-\Sigma^{-1}H\Sigma^{-1})\xi).
\eeq
Hence, the unique critical point $\Sigma_*$ of $\wt{E}$ satisfies 
\bea
F&=&\Sigma_*^{-1}H\Sigma_*^{-1}\nonumber \\
\Rightarrow\quad (F\Sigma_*)^2&=&FH\nonumber \\
\Rightarrow\quad \Sigma_*&=& (F^{1/2})^{-1}H^{1/2}
\eea
where the matrix square root function $\SPD_3\ra\SPD_3$,
$M\mapsto M^{1/2}$, is defined spectrally. Then
\beq
g_*=\frac{\Sigma_*}{\det\Sigma_*}=\left(\frac{\det F}{\det H}\right)^{1/2}F^{-1/2}H^{1/2}.
\eeq

We henceforth set $\Xi=0$. In the case $V\neq 0$, of primary interest, we have not been able to solve for the minimum of $E(g)$ explicitly. Instead, we resort to a numerical method described in the next section.

\section{The numerical method}\label{sec:numerics}\label{sec:anf}

We now return to the problem of primary interest: to minimize $E(\phi,g)$ among all smooth maps $\T^3\ra\SU(2)$ of fixed degree $B$, and all flat metrics $g\in\SPD_3$. Our numerical
scheme is similar to ones introduced in \cite{Leask_2022, babspewin} and based on the idea of arrested Newton flow. For fixed $\phi$, we interpret $E$ as a potential energy on the manifold $\SPD_3$ and solve Newton's law of motion
\begin{equation}\label{g nf}
\ddot g = -\grad_g E(\phi,g)
\end{equation}
with initial data $g(0)=g_0$.  This solution begins to run ``downhill'' in $\SPD_3$. We monitor $E(\phi,g(t))$ and, at any time $t_*$ where $\frac{d}{dt}{E}(\phi,g(t))>0$ we arrest the flow, that is, stop and restart it at the current position but with velocity $0$. The flow converges to the unique minimising metric $g_\phi$.  We minimise over $\phi$ by a similar technique, solving
\begin{equation}\label{phi nf}
\ddot \phi = -\grad_\phi E(\phi,g_\phi)
\end{equation}
with initial data $\phi(0)=\phi_0$ (here $\grad_\phi E$ is the derivative in the first argument, $\phi$, with $g$ treated as constant).  Again, we arrest the flow if $E$ starts to increase.

In practice, we discretize space, placing $\phi$ on a cubic grid of
$N^3$ points with lattice spacing $h=1/N$ and periodic boundary conditions. We replace the spatial derivatives of $\phi$ occurring in $E$ by finite difference approximations. This reduces the \eqref{phi nf} to a system of ODEs in $\R^{4N^3}$, which we then solve using a 4th order Runge-Kutta scheme with fixed time step $\delta t$. The arresting criterion is that
$E(t+\delta t)> E(t)$. After each iteration of the Runge-Kutta scheme, the metric $g_\phi$ is recalculated in a similar way by solving the ODE \eqref{g nf}, using the metric $g_\phi$ from the previous iteration as initial datum.  Once the flow for $g$ has converged, the next iteration of $\phi$ is calculated, and so on.  Each flow is  deemed to have converged to a static solution if the sup norm of $\grad E$ falls below some tolerance. The numerical results presented hereafter were obtained with $N=201$.  The time steps $\delta$ used were 0.0017 for the flow in $\phi$ and 0.1 for the flow in $g$.  The tolerances were $10^{-5}$ for the flow in $\phi$ and $10^{-7}$ for the flow in $g$.

Implementing the method also entails choosing (formal) Riemannian metrics on $\SPD_3$ and $C^\infty(\T^3,SU(2))$. These are required to make sense of both $\grad E$ and the connexions
$\nabla$ occurring implicitly in \eqref{g nf} and \eqref{phi nf}. On $\SPD_3$ we choose the Euclidean metric induced by identifying $\SPD_3$ as a subset of $\R^9$ in the obvious way, that is
\beq
(\xi_1,\xi_2)_{\SPD_3}=\Tr(\xi_1^T\xi_2).
\eeq
Note that this differs from the metric $G$ used in section \ref{sec:eumm}; it is simpler for the current purpose.
On $C^\infty(\T^3,SU(2))$ we choose the $L^2$ metric defined by the volume form $\vol_0$,
\beq
\ip{\eta_1,\eta_2}_{L^2}=\int_{\T^3}h(\eta_1,\eta_2)\vol_0,
\eeq
which is independent of $g$.

For the purposes of numerics, it is convenient to use the sigma model formulation of the model, that is, exploit the isometry between
$(\SU(2),h)$ and $S^3$ with its round metric of unit radius. So, we identify
\beq
\SU(2)\ni\left(\begin{array}{cc}
\phi_0+i\phi_3 & i\phi_1+\phi_2 \\
i\phi_1-\phi_2 & \phi_0-i\phi_3\end{array}\right)
\leftrightarrow(\phi_0,\phi_1,\phi_2,\phi_3)\in S^3
\eeq
and interpret the Skyrme field as a unit length vector in Euclidean $\R^4$. The terms occurring in the Skyrme energy are easily converted to this formulation,
\bea
    E_2(\phi)&=&\int_{\T^3} g^{ij} (\cd_i\phi\cdot\cd_j\phi) \sqrt{|g|}\vol_0\label{pion E2}\\
    E_4(\phi)&=&\frac12\int_{\T^3} (g^{ij}g^{kl}-g^{il}g^{jk}) (\cd_i\phi\cdot\cd_j\phi) (\cd_k\phi\cdot\cd_l\phi)
    \sqrt{|g|}\vol_0\label{pion E4}\\
    E_0(\phi)&=&2m_\pi^2\int_{\T^3}(1-\phi_0)\sqrt{|g|}\vol_0\label{pion E0},
\eea
as are the matrices defined in section \ref{sec:eumm},
\bea
    H_{ij}&=&\int_{\T^3}(\cd_i\phi\cdot\cd_j\phi)\vol_0\\
    F_{ij}&=&\frac{1}{4}\eps_{ikl}\eps_{jmn}\int_{\T^3}\left\{(\cd_k\phi\cdot\cd_m\phi)(\cd_l\phi\cdot\cd_n\phi) - (\cd_k\phi\cdot\cd_n\phi)(\cd_l\phi\cdot\cd_m\phi) \right\}\vol_0.
\eea
Now, the gradient $E=E_2+E_4+E_0$, regarded as a function on $C^\infty(\T^3,S^3)$, is
\beq
(\grad_\phi E)(\phi,g)_{C^\infty}=
\sqrt{|g|}P_\phi\left(-2g^{ij}\cd_i\cd_j\phi
-2(g^{li}g^{jk}-g^{ij}g^{kl})\cd_i(\cd_j\phi\cdot\cd_k\phi\, \cd_l\phi)-
2m_\pi^2\evec_0\right)
\eeq
where $\evec_0=(1,0,0,0)$ and $P_\phi:\R^4\ra\R^4$ is the projector orthogonal to $\phi$, that is 
\beq
P_\phi(\vvec)=\vvec-(\phi\cdot\vvec)\phi.
\eeq 
 
The gradient of $E$ regarded as a function on $\SPD_3$ is easily deduced from \eqref{scgii} and the diffeomorphism $\sigma:\SPD_3\ra\SPD_3$. We note that
\beq
\d\wt{E}_\Sigma(\xi)=\Tr((F-\Sigma^{-1}H\Sigma^{-1}
-\frac{C_0\Sigma^{-1}}{\det\Sigma})\xi),
\eeq
and so
\beq
(\grad\wt{E})(\Sigma)=F-\Sigma^{-1}H\Sigma^{-1}
-\frac{C_0\Sigma^{-1}}{\det\Sigma}.
\eeq
Furthermore
\beq
\d\sigma_{g}(\xi)=\frac{1}{\sqrt{\det g}}\left(\xi-\frac12\Tr(g^{-1}\xi)g\right),
\eeq
whence
\bea
\d E_g(\xi)&=&\d\wt{E}_{\sigma(g)}\circ\d\sigma_g(\xi)\nonumber\\
&=&\Tr\left\{\grad\wt{E}(|g|^{-1/2}g)|g|^{-1/2}(\xi-\frac12
\Tr(g^{-1}\xi)g)\right\}\nonumber \\
&=&\Tr\left\{|g|^{-1/2}\left[
\grad\wt{E}(|g|^{-1/2}g)-\frac12\Tr(g\grad\wt{E}(|g|^{-1/2}g))g^{-1}\right]\xi\right\},
\eea
and so
\beq
(\grad_g E)(\phi,g)_{\SPD_3}=|g|^{-1/2}\left[
\grad\wt{E}(|g|^{-1/2}g)-\frac12\Tr(g\grad\wt{E}(|g|^{-1/2}g))g^{-1}\right].
\eeq


\section{Skyrme crystal solutions}
\label{sec: Skyrme crystal solutions}

This section presents the results of the numerical scheme just described in the charge $B=4$ sector, concentrating on the model with pion mass $m_\pi=1$. Our approach is to treat the pion mass as a continuous variable parameter $m_\pi=t\geq 0$: we minimize the energy
\beq
E_{(t)}(\phi,g)=E_2(\phi,g)+E_4(\phi,g)+t\int_{\T^3}(1-\phi_0)\vol_g
\eeq
starting in the massless case $t=0$, and then increasing $t$ gradually to $1$.

We begin by reviewing the lowest energy solution known in the massless case, $t=0$, found independently by Kugler and Shtrikman \cite{Kugler_1988,Kugler_1989}  and Castillejo \textit{et al.} \cite{Castillejo_1989}. This can be found by starting with the initial field \cite{Castillejo_1989}
\begin{equation}\label{1/2}
    \phi_0  = - c_1 c_2 c_3, \quad \phi_1 = s_1 \sqrt{1-\frac{s_2^2}{2}-\frac{s_3^2}{2}+\frac{s_2^2 s_3^2}{3}},
\end{equation}
$\phi_2,\phi_3$ obtained by cyclic permutation, where $s_i=\sin 2\pi x_i$ and
$c_i=\cos 2\pi x_i$, and initial metric $g=\I_3$. This quickly converges under arrested Newton flow to a solution with $\phi$
close to \eqref{1/2} and $g=L^2\I_3$ with $L=4.61$, corresponding to a cubic period lattice of side length $L$. This solution is depicted in figure \ref{fig: 1/2-crystal}. It represents a simple cubic lattice of half-skyrmions. That is, $\phi$ maps each of the eight subcubes of side length $L/2$ to either the upper ($\phi_0\geq0$) or lower ($\phi_0\leq 0$) hemisphere of $S^3$, contributing charge $B=1/2$ to the total topological charge of the unit cell. For this reason, we refer to this solution as the $1/2$-crystal and denote it $(\phi_{1/2},g_{1/2})$.

\begin{figure}[t]
    \centering
    \begin{subfigure}[b]{0.4\textwidth}
        \centering
        \includegraphics[width=\textwidth]{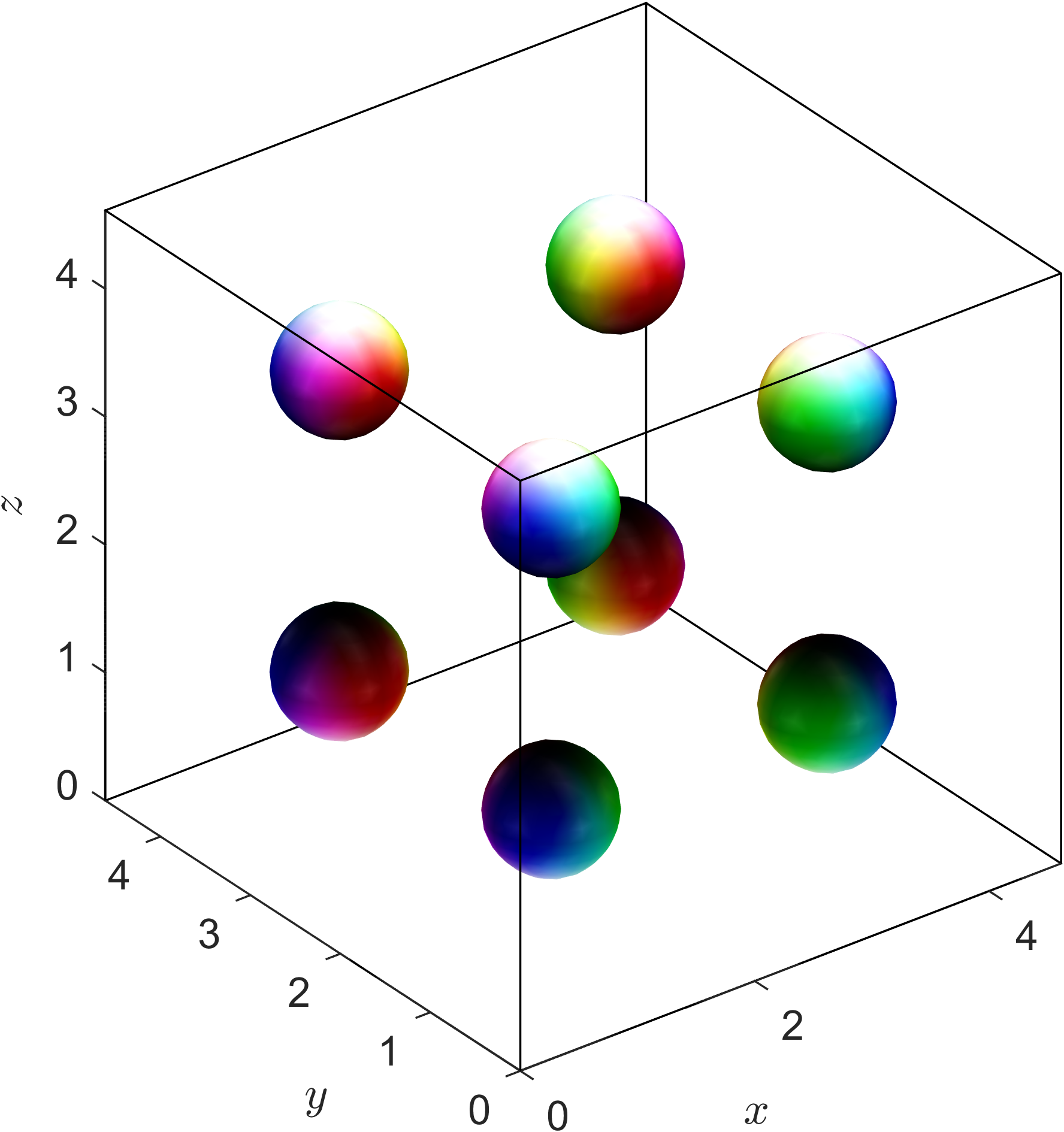}
        \caption{}
        \label{fig: Massless 1/2-crystal BD} 
    \end{subfigure}
    ~
    \begin{subfigure}[b]{0.4\textwidth}
        \centering
        \includegraphics[width=\textwidth]{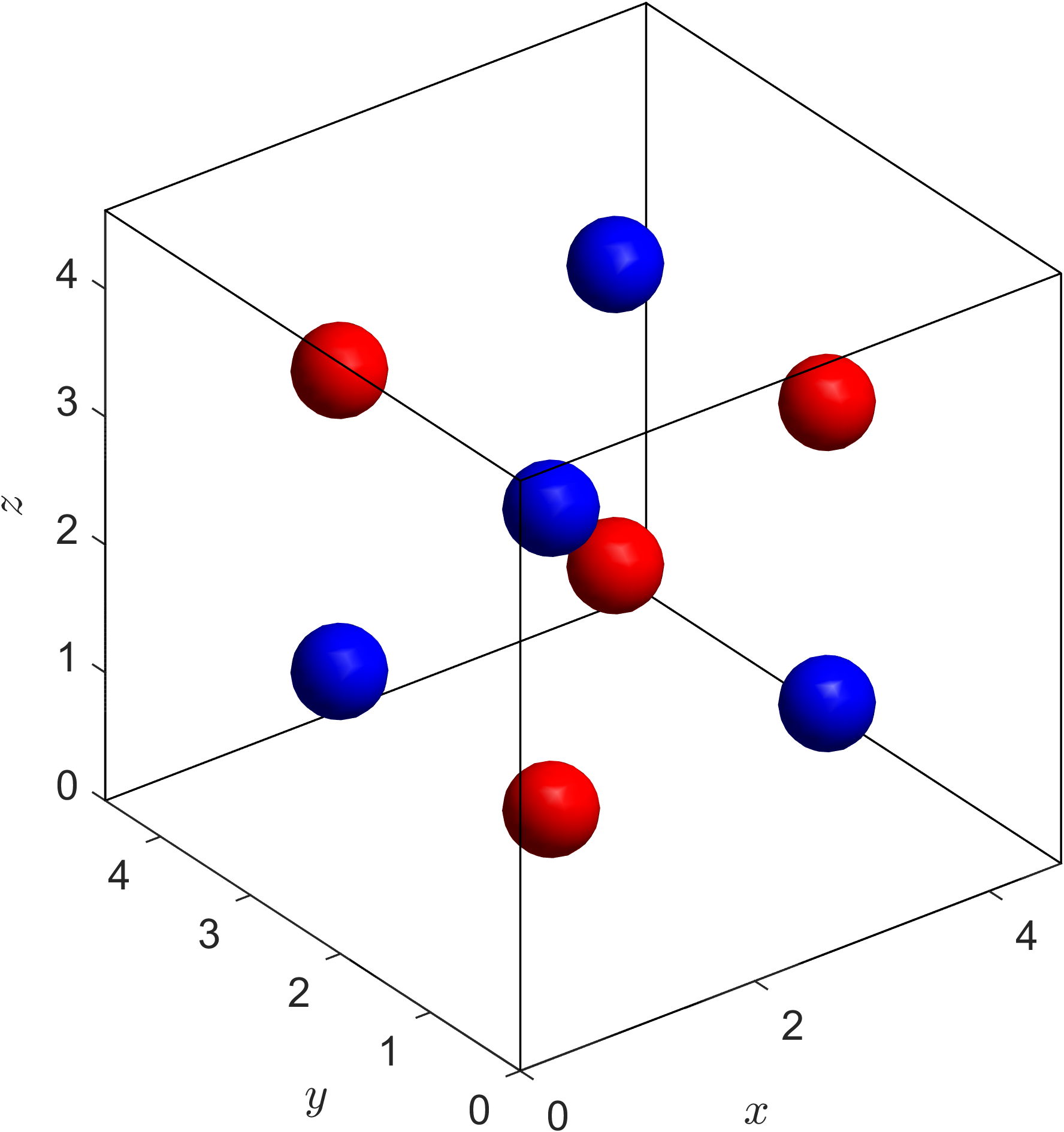}
        \caption{}
        \label{fig: Massless 1/2-crystal SD}
    \end{subfigure}
    \caption{The $1/2$-crystal solution of the massless Skyrme model. The baryon density is depicted in \ref{fig: Massless 1/2-crystal BD} and isosurface plots of the $\phi_0$ field, where $\textcolor{red}{\phi_0=0.9}$ and $\textcolor{blue}{\phi_0=-0.9}$, are shown in \ref{fig: Massless 1/2-crystal SD}.}
    \label{fig: 1/2-crystal}
\end{figure}

It is important to note that $E_{(0)}=E_2+E_4$ is invariant under the natural action of $\SO(4)$ on the target three-sphere, that is, for all $(\phi,g)\in\MM$ and all $R\in \SO(4)$,
$E_{(0)}(R\phi,g)=E_{(0)}(\phi,g)$. Hence the solution described above is just one
critical point of $E_{(0)}$ lying in a $6$-dimensional family of critical points, its orbit under $\SO(4)$. If we now switch on the pion mass, that is, consider $E_{(t)}$
for small $t>0$, we may ask which (if any) of these critical points survive the perturbation. It is useful to switch perspective slightly: rather than fixing the perturbation and considering what happens to all points in the $\SO(4)$ orbit of the $1/2$-crystal, it is convenient to fix the field and metric to be the $1/2$-crystal, and consider what happens to this fixed configuration under the $\SO(4)$ orbit of the perturbation. That is, we ask for which $p\in S^3$, if any, does the $1/2$-crystal $(\phi_{1/2},g_{1/2})$ lie in a curve 
$(\phi(t),g(t))$ of critical points of the $t$-parametrized family of functions
\beq
E_{(t)}^p(\phi,g)=E_2(\phi,g)+E_4(\phi,g)+t\int_{\T^3}(1-p\cdot\phi)\vol_g.
\eeq
(We recover the original function $E_{(t)}$ by choosing $p=(1,0,0,0)$.)
To answer this question, we will need to understand the symmetries of the $1/2$-crystal in some detail.  

Denote by $O$ the subgroup of $\SO(3)$ consisting of matrices that map the integer lattice $\Z^3$ to itself. This is a finite group of order $24$, the rotational symmetries of the cube. The manifold $\T^3$ is itself an abelian Lie group whose group operation is translation modulo $\Z^3$. The semidirect product $\Aut(\T^3)=O\ltimes \T^3$
\beq
(A_1,\vvec_1)\bullet(A_2,\vvec_2)=(A_1A_2,A_1\vvec_2+\vvec_1)
\eeq
acts on $\T^3$ by $(A,\vvec):\xvec\mapsto A\xvec+\vvec$. This induces a right action of $\Aut(\T^3)$ on $\MM=C^\infty(\T^3,S^3)\times\SPD_3$ by $S:(\phi,g)\mapsto (\phi\circ S,
S^*g)$ or, more explicitly,
\beq
(A,\vvec)\cdot(\phi(\xvec),g)=(\phi(A(\xvec+\vvec)),A^TgA).
\eeq

Having completed these preliminaries, we observe that the 
energy of the massless Skyrme model $E_{(0)}:\MM\ra\R$ is invariant under the left action of $G=\SO(4)\times \Aut(\T^3)$ on $\MM$,
\beq
(R,S)\cdot(\phi,g)=(R\circ \phi\circ S^{-1}, (S^{-1})^*g).
\eeq
The $1/2$-crystal $(\phi_{1/2},g_{1/2})$ is a critical point (in fact a minimum) of $E_{(0)}$. Its stabilizer $\Gamma$ (the subgroup of $\SO(4)\times\Aut(\T^3)$ that leaves it fixed) is an order $192$ group generated by
\beq
\begin{array}{ll}
R(\phi)=(\phi_0,\phi_2,\phi_3,\phi_1), &
S(\xvec)=(x_2,x_3,x_1),\\
R(\phi)=(\phi_0,\phi_2,-\phi_1,\phi_3), &
S(\xvec)=(x_2,-x_1,x_3),\\
R(\phi)=(-\phi_0,-\phi_1,\phi_2,\phi_3), &
S(\xvec)=(x_1+\frac12,x_2,x_3).
\end{array}
\eeq
The image of the natural projection $\pi:\Gamma\to SO(4)$ is naturally isomorphic to the octahedral group $O_h$, and the kernel is isomorphic to $\mathbb{Z}_2\times\mathbb{Z}_2$.

Once we turn on the perturbation, the symmetry group of the energy function
$E_{(t)}^p$ is broken to $\SO(3)_p\times\Aut(\T^3)$ where
\beq
\SO(3)_p=\{R\in \SO(4): Rp=p\}.
\eeq
Let us define the reduced stabilizer of the $1/2$-crystal to be
\beq
\Gamma_p=\Gamma\cap(\SO(3)_p\times\Aut(\T^3)),
\eeq
and the set of fixed points of $\Gamma_p$ in $\MM$ to be 
\beq
\MM^{\Gamma_p}=\{(\phi,g)\in\MM: \forall q\in\Gamma_p,\: q\cdot(\phi,g)=(\phi,g)\}.
\eeq
Formally, this is a submanifold of $\MM$, and it contains $(\phi_{1/2},g_{1/2})$ for all $p$, by construction. 
 By the Principle of Symmetric criticality, a point
$(\phi,g)\in \MM^{\Gamma_p}$ is a critical point of $E_{(t)}^p$ if (and only if) it is a critical point of its restriction $E_{(t)}^p|:\MM^{\Gamma_p}\ra\R$.  
For generic choices of $p\in S^3$ we expect $\Gamma_p$ to be trivial, so that 
$\MM^{\Gamma_p}=\MM$, and this observation confers no advantage. The interesting case is when the intersection of $\MM^{\Gamma_p}$ with the $G$ orbit of $(\phi_{1/2},g_{1/2})$ is (locally) just $(\phi_{1/2},g_{1/2})$. Then $(\phi_{1/2},g_{1/2})$ is an {\em isolated} critical point of $E_{(0)}|:\MM^{\Gamma_p}\ra\R$. If, as seems likely, it is also a {\em nondegenerate}
critical point of $E_{(0)}^p|$ (meaning that the Hessian of $E_{(0)}^p|$ at $(\phi_{1/2},g_{1/2})$ is nondegenerate), then the persistence of a critical point for $t>0$ sufficiently small follows from the Inverse Function Theorem
applied to $\d E_{(t)}^p|$. That is, there exists $\eps>0$ and a (unique)
smooth curve $\gamma:(-\eps,\eps)\ra\MM^{\Gamma_p}$ such that $\gamma(0)=(\phi_{1/2},g_{1/2})$ and $\d E_{(t)}^p|_{\gamma(t)}=0$ for all $t\in (-\eps,\eps)$.

To summarize, we expect $(\phi_{1/2},g_{1/2})$ to smoothly deform into a critical point of $E_{(t)}^p$ (as $t$ increases from $0$) if $p$ is chosen so that a neighbourhood of $(\phi_{1/2},g_{1/2})$ in $\MM^{\Gamma_p}$ intersects the $G$ orbit of $(\phi_{1/2},g_{1/2})$ only at $(\phi_{1/2},g_{1/2})$. Let us call this condition the {\em isolation condition}. The next task is to understand this condition on $p$ at an algebraic level. 

Assume that $p$ fails the isolation condition. Then there exists a regular curve 
$q:(-\eps,\eps)\ra G$ with $q(0)=e$ such that, for all $t$, $q(t)\cdot(\phi_{1/2},g_{1/2})\in\MM^{\Gamma_p}$ or, more explicitly, 
for all $Q\in\Gamma_p$, and $t\in(-\eps,\eps)$
\bea
Q\cdot q(t) \cdot (\phi_{1/2},g_{1/2})&=&q(t)\cdot(\phi_{1/2},g_{1/2})\nonumber\\
\Rightarrow\quad[q(t)^{-1}Qq(t)]\cdot(\phi_{1/2},g_{1/2})&=&(\phi_{1/2},g_{1/2}).
\eea
Hence, for all $Q\in\Gamma_p$ and $t$, $q(t)^{-1}Qq(t)\in\Gamma$. But $\Gamma$ is discrete (in fact, finite), so for all $t$ and $Q$, 
\bea
q(t)^{-1}Qq(t)&=&q(0)^{-1}Qq(0)=Q\\
\Rightarrow\quad
Qq(t)Q^{-1}&=&q(t).
\eea
The derivative of this equation at $t=0$ implies that there exists
some nonzero $\xi\in\g$ (the Lie algebra of $G$), namely $\xi=\dot{q}(0)$,
such that $Ad_Q\xi=\xi$. Conversely, given a nonzero $\xi\in\g$ such that
$Ad_Q\xi=\xi$ for all $Q\in\Gamma_p$, we can construct a curve $\gamma(t)=\exp(t\xi)$ such that $\gamma(t)\cdot(\phi_{1/2},g_{1/2})$ remains in $\MM^{\Gamma_p}$. Hence, the isolation condition is that, for all $\xi\in\g\less\{0\}$, there exists some $Q\in\Gamma_p$ such that $Ad_Q\xi\neq\xi$. More succinctly: $p$ satisfies the isolation condition if and only if the
adjoint representation of $\Gamma_p$ on $\g$ contains no copies of the trivial representation. 

This reduces the problem to one in the representation theory of subgroups of $O_h$.
Given a subgroup $H$ of $O_h\subset \SO(4)$, we determine whether its action on $\mathbb{R}^4$ contains copies of the trivial representation.
If not, it cannot arise as $\pi(\Gamma_p)$ for any choice of $p$.
If it does, $\pi^{-1}(H)$ is a candidate for $\Gamma_p$ for any $p$ in a one-dimensional invariant subspace of the action.
This produces a short list of candidate subgroups.
For each of these we count copies of the trivial representation in the adjoint representation of $\pi^{-1}(H)$ on $\g$.
If there are none, this is a candidate for $\Gamma_p$ for $p$ satisfying the isolation condition.

The results are summarized in table \ref{table1}. We find 28 points $p$ for which $(\phi_{1/2},g_{1/2})$ is an isolated critical point of $E_{(0)}^p$ in
$\MM^{\Gamma_p}$, falling into 4 distinct classes. One class is $p\in\{(1,0,0,0),(-1,0,0,0)\}$. The other three classes all have $p_0=0$ and hence $(p_1,p_2,p_3)\in S^2\subset\R^3$, pointing along some symmetry line of the unit cube: towards the centre of a face (e.g.\ $p=(0,0,0,1)$), the centre of an edge (e.g.\ $p=(0,1,1,0)/\sqrt{2}$) or a vertex (e.g.\ $p=(0,1,1,1)/\sqrt{3}$).

\begin{table}
\begin{center}
\begin{tabular}{|l|l|l|l|l|l|}\hline
$p$ & $|\Gamma_p|$ & $\pi(\Gamma_p)$ & Description as subgroup of $O_h$ & $g$ & label\\
\hline
$(1,0,0,0)$ & 96 & $O$ & Orientation preserving & $\diag(a,a,a)$ & $1/2$-crystal\\
$(0,0,0,1)$ & 32 & $C_{4v}$ & Maps a face to itself & $\diag(a,a,b)$ & sheet \\
$(0,0,1,1)/\sqrt{2}$ & 16 & $C_{2v}$ & Maps an edge to itself & $\diag(a,b,b)$ & chain\\
$(0,1,1,1)/\sqrt{3}$ & 24 & $C_{3v}$ & Maps a vertex to itself & $ \diag(a,a,a)$ & $\alpha$-crystal\\
\hline
\end{tabular}
\end{center}

\caption{Points $p\in S^3$ for which the $1/2$-crystal is isolated in $\MM^{\Gamma_p}$, and hence is expected to continue to a critical point of the massive Skyrme model. The leftmost column gives one representative point in each class. Subsequent columns record the order of the corresponding stabilizer $\Gamma_p\subset \Gamma$, the image of $\Gamma_p$ in $\pi(\Gamma)=O_h$, its description as a subgroup of the group of symmetries of the cube, the most general metric consistent with the symmetry, and a descriptive label of the corresponding crystal.} \label{table1}
\end{table}

To do numerics, we switch back to the viewpoint of internally rotating the
field $\phi_{1/2}$, rather than the energy functional, that is, we set $p=(1,0,0,0)$ in $E_{(t)}^p$ and start with the configuration
\beq
\phi=Q\phi_{1/2},\qquad g=g_{1/2},
\eeq
where $Q$ is any $\SO(4)$ matrix whose top row is $p$ (the inverse of an $\SO(4)$ matrix mapping $(1,0,0,0)\mapsto p$). We then minimize $E_{(t)}$ using arrested Newton flow for a sequence of pion masses $m_\pi=t$ starting at $t=0$ and ending at $t=1$. As expected each of the 4 types of critical point smoothly continues. Somewhat unexpectedly, they are all, as far as we can determine, local minima of $E_{(t)}$; none are saddle points. We have checked this by
perturbing the solutions with random perturbations breaking all symmetries, finding that they always relax back to the solutions presented. 

\begin{figure}[t]
    \centering
    \begin{subfigure}[b]{0.23\textwidth}
        \centering
        \includegraphics[width=\textwidth]{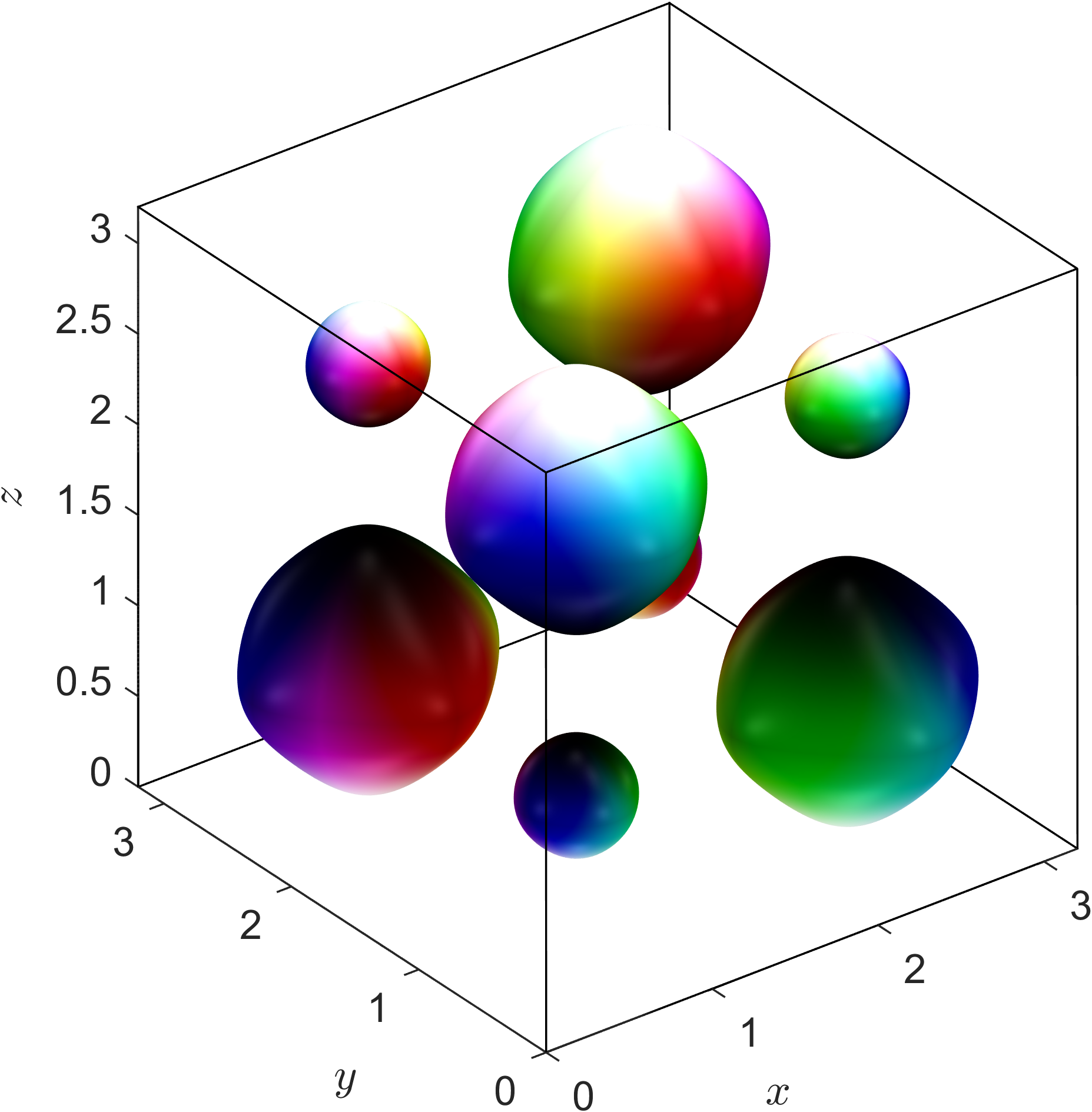}
    \end{subfigure}
    ~
    \begin{subfigure}[b]{0.23\textwidth}
        \centering
        \includegraphics[width=\textwidth]{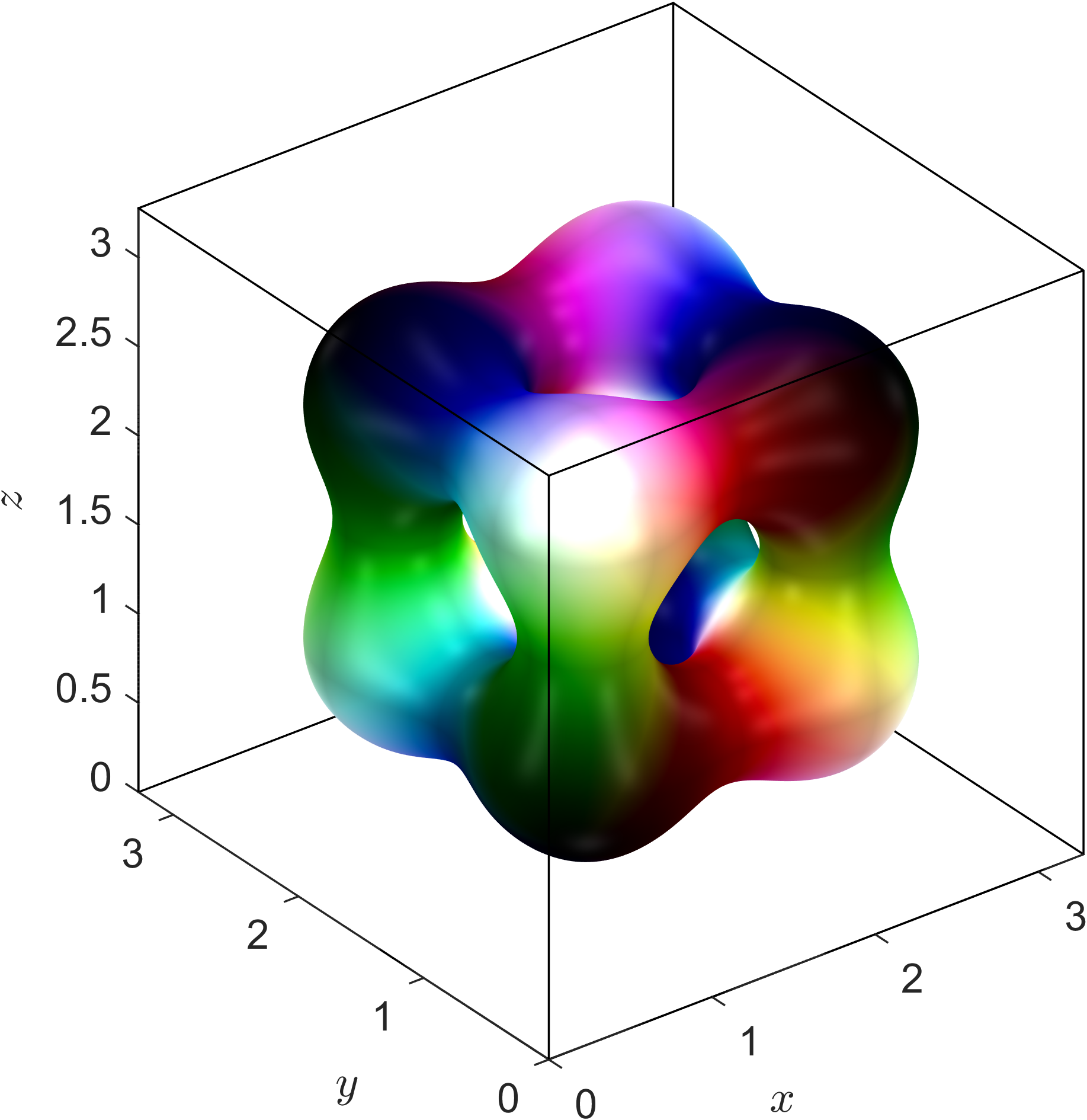}
    \end{subfigure}
    ~
    \begin{subfigure}[b]{0.23\textwidth}
        \centering
        \includegraphics[width=\textwidth]{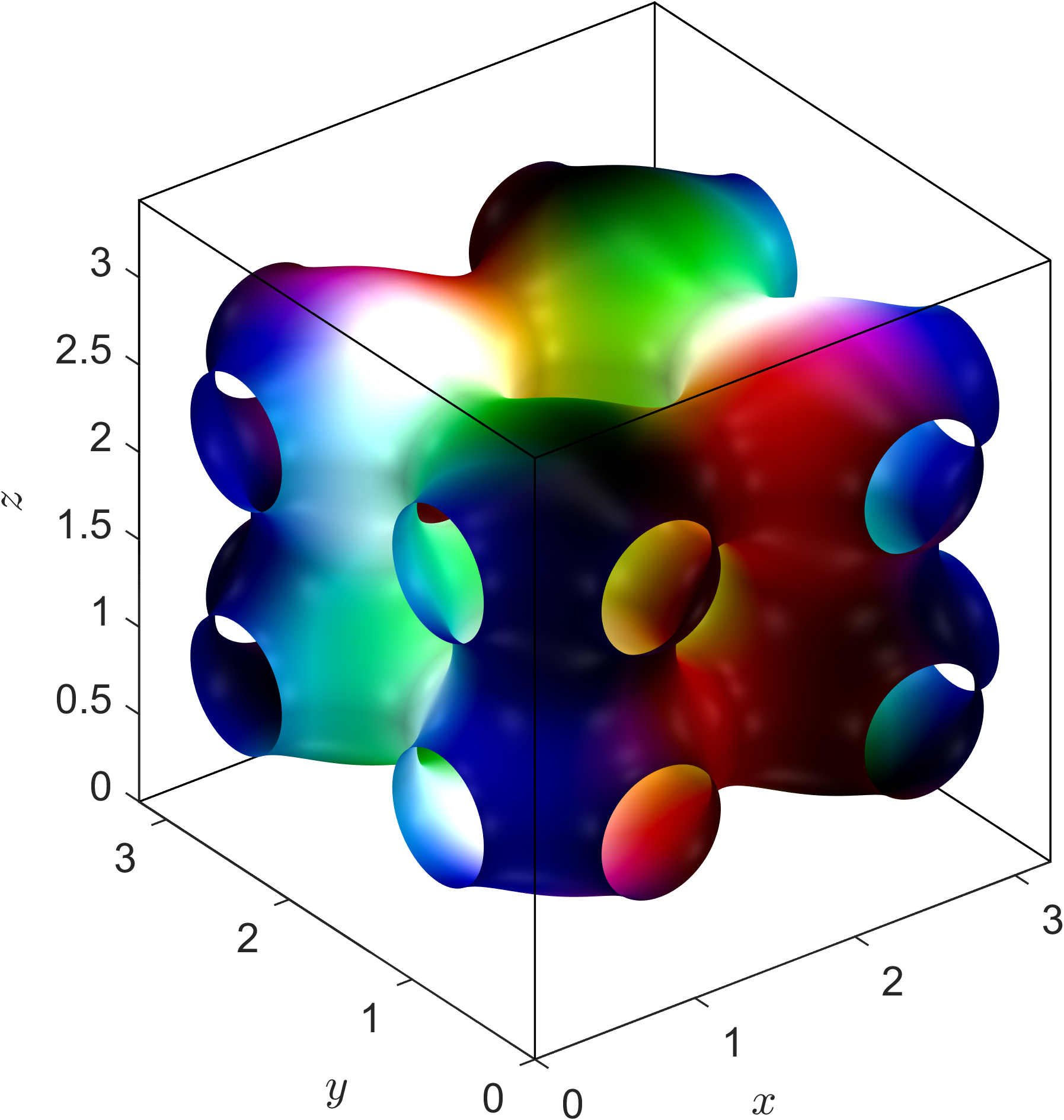}
    \end{subfigure}
    ~
    \begin{subfigure}[b]{0.23\textwidth}
        \centering
        \includegraphics[width=\textwidth]{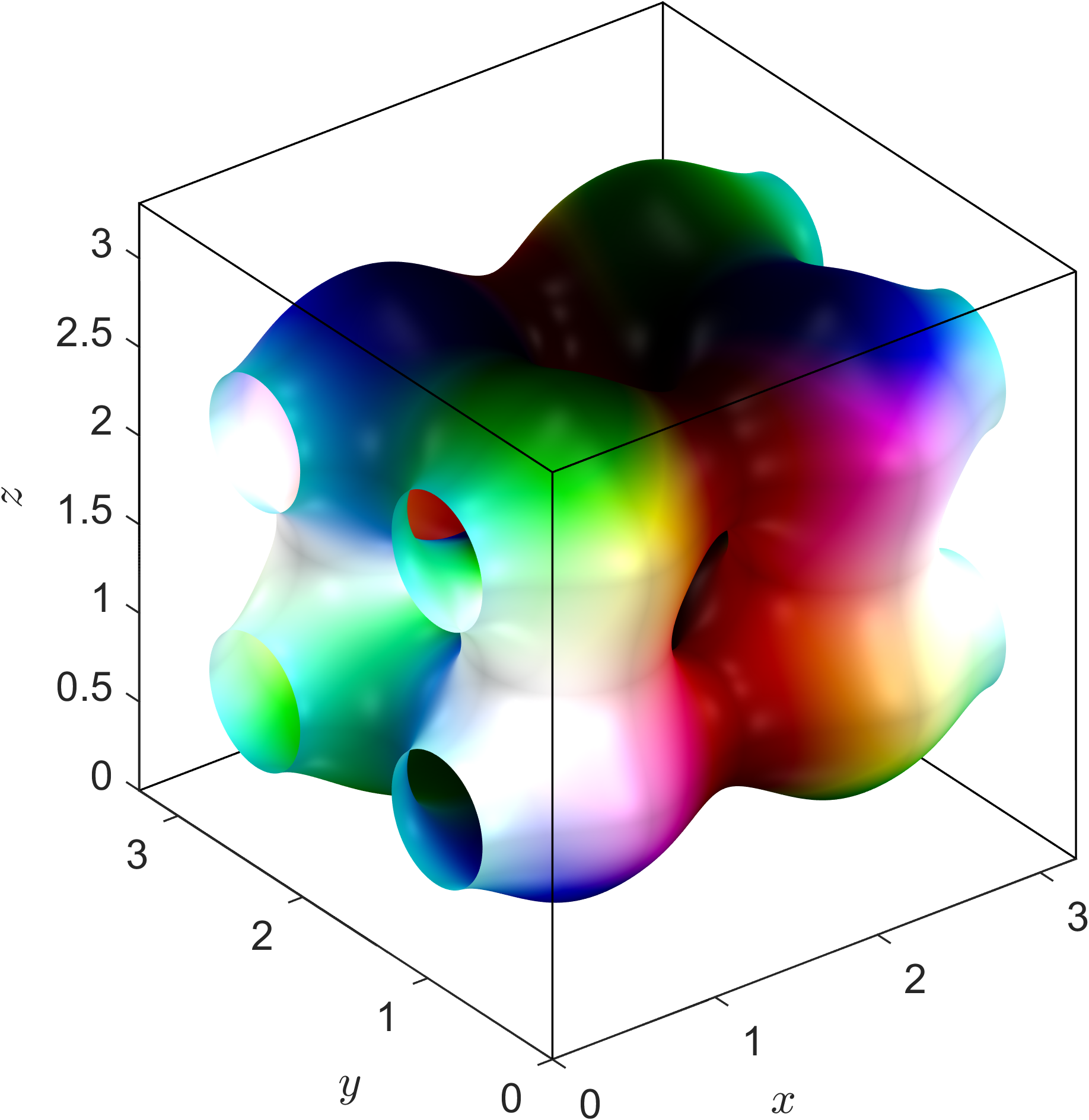}
    \end{subfigure}
    \\
    \begin{subfigure}[b]{0.23\textwidth}
        \centering
        \includegraphics[width=\textwidth]{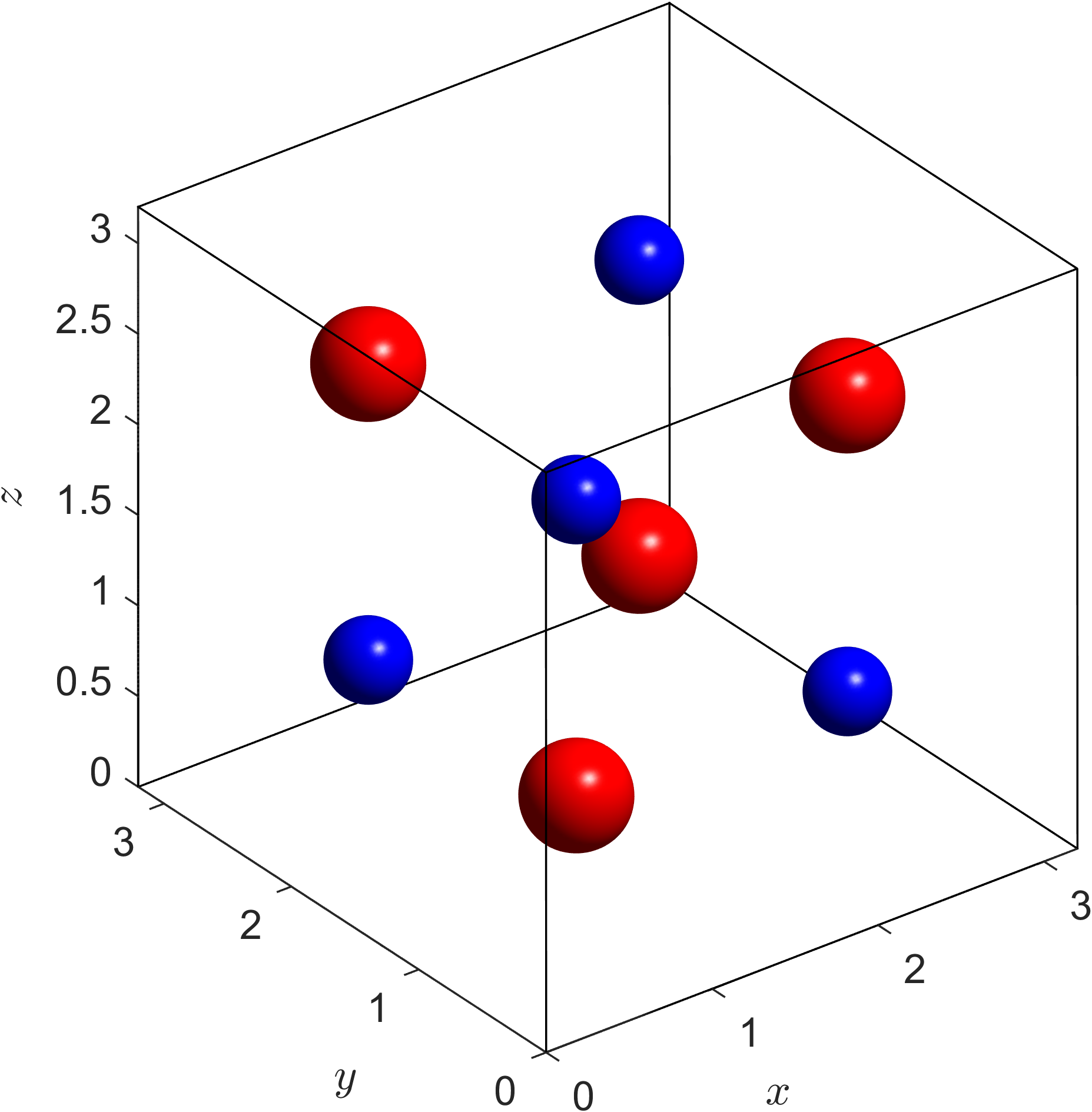}
        \caption{$1/2$-crystal}
    \end{subfigure}
    ~
    \begin{subfigure}[b]{0.23\textwidth}
        \centering
        \includegraphics[width=\textwidth]{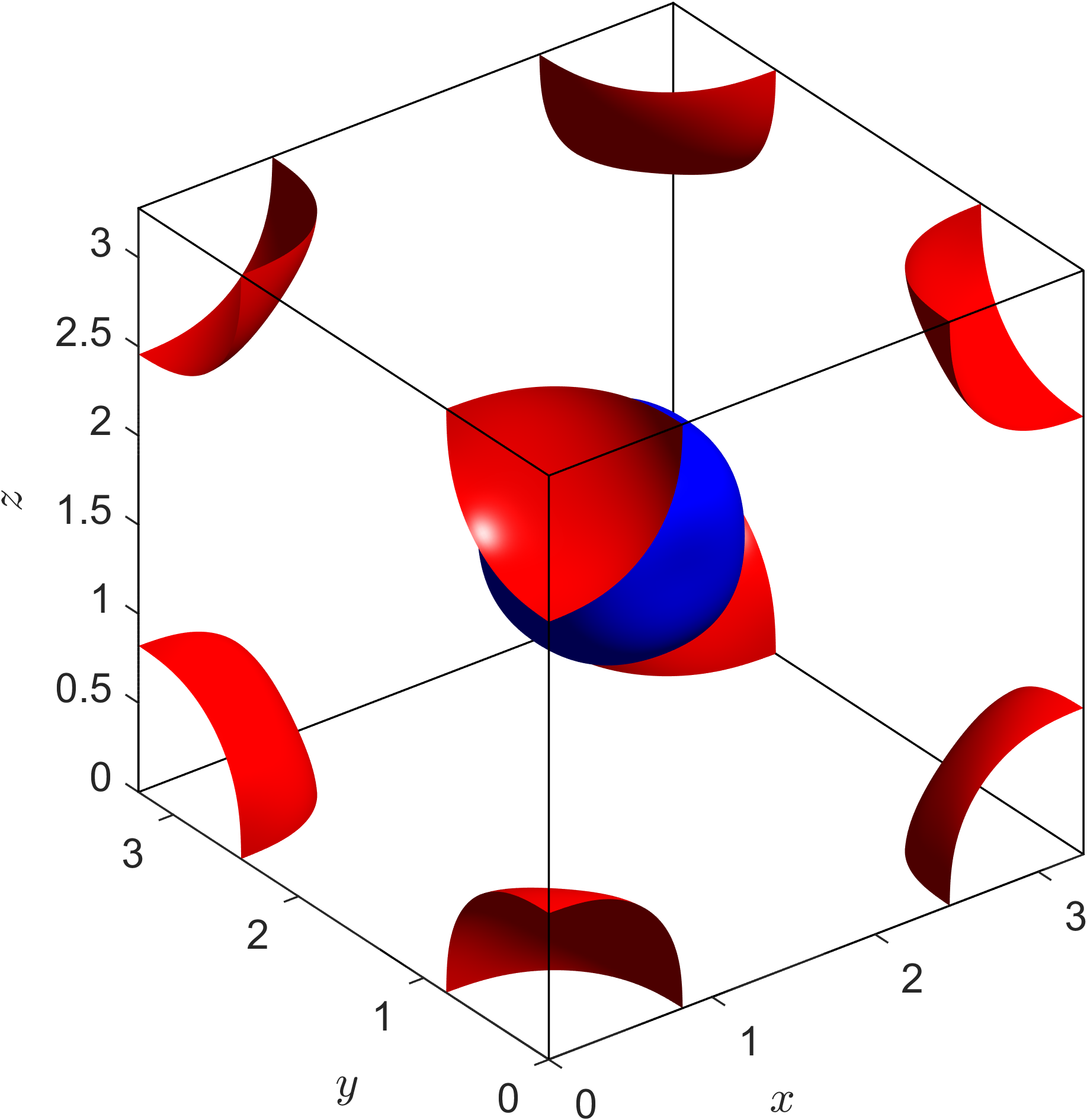}
        \caption{$\alpha$-crystal}
    \end{subfigure}
    ~
    \begin{subfigure}[b]{0.23\textwidth}
        \centering
        \includegraphics[width=\textwidth]{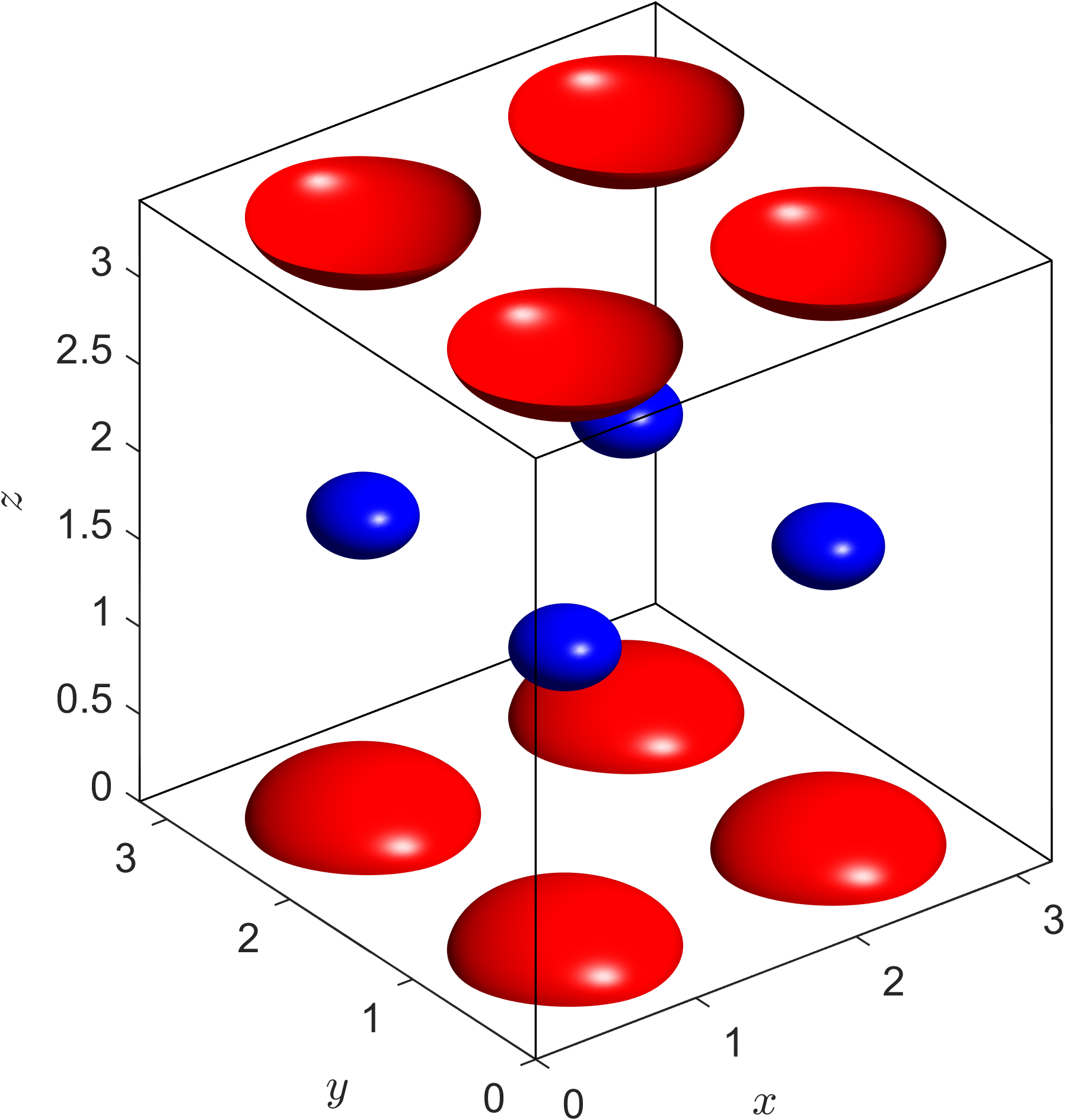}
        \caption{sheet-crystal}
    \end{subfigure}
    ~
    \begin{subfigure}[b]{0.23\textwidth}
        \centering
        \includegraphics[width=\textwidth]{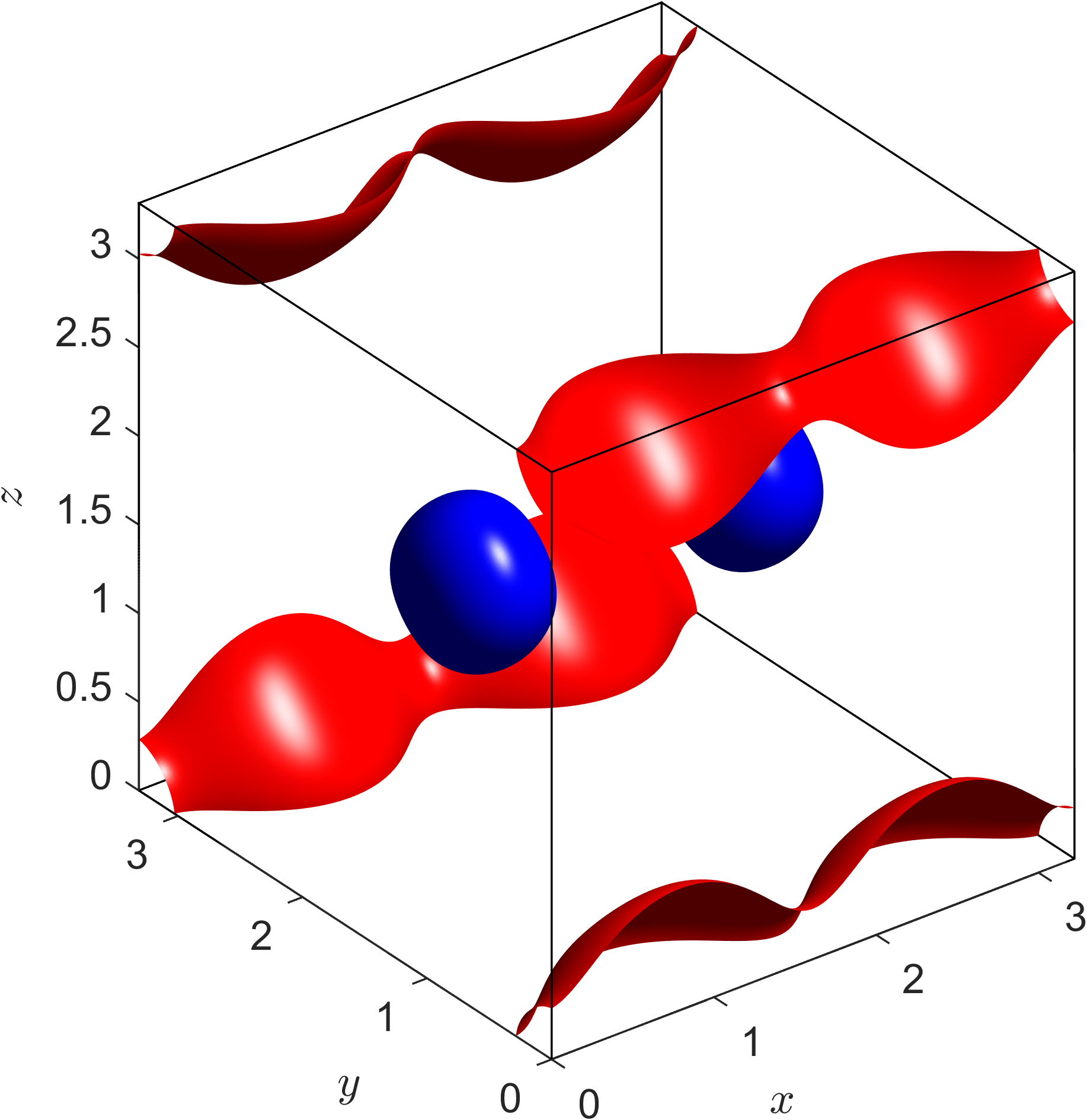}
        \caption{chain-crystal}
    \end{subfigure}
    \caption{Skyrme crystals in the model with pion mass $m_\pi=1$. The top row is the isosurface plots of the baryon density. The bottom row is isosurface plots of the $\phi_0$ field, where $\textcolor{red}{\phi_0=0.9}$ and $\textcolor{blue}{\phi_0=-0.9}$.}
    \label{fig: Massive Skyrme crystals}
\end{figure}

The solutions at pion mass $m_\pi=1$ are depicted in figure \ref{fig: Massive Skyrme crystals}, labelled as in the final column of table \ref{table1}. Ordered by energy, we find sheet $<$ chain $<$ $\alpha$-crystal $<$ $1/2$-crystal, though the chain and $\alpha$-crystals are so close in energy that their order is somewhat uncertain. The energies per baryon per unit cell are 
\begin{equation}
\begin{aligned}
\frac{E_{1/2}}{B}&= 1.2417\times 12\pi^2 = 147.058,\\
\frac{E_\alpha}{B}&= 1.2368 \times 12\pi^2 = 146.479\\
\frac{E_\textrm{chain}}{B}&= 1.2368 \times 12\pi^2 = 146.479\\
\frac{E_\textrm{sheet}}{B}&= 1.2365 \times 12\pi^2 = 146.451.
\end{aligned}
\end{equation}
Neither the sheet-crystal  nor the chain-crystal has an isotropic metric, meaning these crystals do {\em not} have a cubic period lattice. The $\alpha$-crystal and the $1/2$-crystal do have cubic period lattices, as is consistent with our symmetry analysis (see column 5 of table \ref{table1}). The minimal metrics are
\begin{equation}
\begin{aligned}
g_{1/2}&=L^2\I_3,& L&=3.202,  \\
g_{\alpha}&=L^2\I_3,& L&=3.278,\\
g_\textrm{chain}&=\diag(L_1^2,L_2^2,L_2^2),& L_1&=3.221,\quad L_2=3.312,\\
g_\textrm{sheet}&=\diag(L_1^2,L_1^2,L_2^2),& L_1&=3.222,\quad L_2=3.442
\end{aligned}
\end{equation}
from which we deduce that the unit cells for the sheet and chain crystals are trigonal (cuboidal with one pair of periods equal), but with opposite types of distortion: the sheet's unit cell is a stretched cube, the chain's a squashed cube.  Interestingly, the ordering of the volumes of the solutions' unit cells is the reverse of the ordering of their energies, with the sheet-crystal occupying the greatest volume and the 1/2-crystal the least.

Restricting the kinetic energy functional of the model to the isospin orbit of a given static solution we obtain a left invariant metric on $\SO(3)$ called the isospin inertia tensor, which is of some significance in the method of
rigid body quantization \cite{Baskerville_1996,Wereszczynski_2022}. The kinetic energy associated with the potential energy $E_2+E_4+E_0$ given in \eqref{pion E2}-\eqref{pion E0} is
\begin{equation}
    T(\phi,\dot\phi) = \int_{\T^3}\big[\dot{\phi}\cdot\dot{\phi} + g^{ij}\left\{(\dot\phi\cdot\dot\phi)\,(\cd_i\phi\cdot\cd_j\phi) - (\dot\phi\cdot\cd_i\phi)\,(\dot\phi\cdot\cd_j\phi) \right\}\big]\sqrt{|g|}\vol_0.
\end{equation}
Writing $\dot\phi=X^iJ_i\phi$, with $J_i$ being the basis for $\mathfrak{so}(3)$ given by
\begin{equation}
J_1 = \begin{pmatrix}0&0&0&0\\0&0&0&0\\0&0&0&-1\\0&0&1&0\end{pmatrix},\quad
J_2 = \begin{pmatrix}0&0&0&0\\0&0&0&1\\0&0&0&0\\0&-1&0&0\end{pmatrix},\quad
J_3 = \begin{pmatrix}0&0&0&0\\0&0&-1&0\\0&1&0&0\\0&0&0&0\end{pmatrix},
\end{equation}
we find that $T(\phi,\dot\phi)=\frac{1}{2}X^iU_{ij}X^j$, where $U$ is the symmetric matrix given by
\begin{multline}
U_{ij} = 2\int_{\T^3}\big[\delta_{ij}\phi_k\phi_k-\phi_i\phi_j + g^{kl}(\delta_{ij}-\phi_i\phi_j)\cd_k\phi_0\cd_l\phi_0 \\
+g^{kl}(\phi_m\phi_m \cd_k\phi_i\cd_l\phi_j +\phi_0\phi_j\partial_k\phi_0\cd_l\phi_i +\phi_0\phi_i\cd_l\phi_0\cd_k\phi_j )\big]\sqrt{|g|}\vol_0
\end{multline}
and repeated indices are summed from 1 to 3.  We find that, except for the $1/2$-crystal, this matrix is {\em not}{ isotropic}:
\begin{equation}
\begin{aligned}
    U^{1/2} =
        \begin{pmatrix}
            165.2 & 0 & 0 \\
            0 & 165.2 & 0 \\
            0 & 0 & 165.2
        \end{pmatrix}, \\
            U^{\alpha} =
        \begin{pmatrix}
            135.5 & 0 & 0 \\
            0 & 135.5 & 0 \\
            0 & 0 & 167.3
        \end{pmatrix} \\
    U^\textrm{chain} =
        \begin{pmatrix}
            135.6 & 0 & 0 \\
            0 & 135.7 & 0 \\
            0 & 0 & 167.2
        \end{pmatrix}, \\
    U^\textrm{sheet} =
        \begin{pmatrix}
            135.8 & 0 & 0 \\
            0 & 135.8 & 0 \\
            0 & 0 & 166.8
        \end{pmatrix}.
\end{aligned}
\end{equation}

As far as we are aware, in addition to the $1/2$-crystal, only the $\alpha$-crystal has been previously determined in the massive Skyrme model \cite{Feist_2012}. Neither of these is the minimal energy crystal.

It is interesting to track the energy as a function of pion mass, see figure \ref{fig: Massive Skyrme crystals - Comparison}. As $m_\pi$ increases, all of the crystals' energies increase relative to that of the one-skyrmion.  This is an indication that classical binding energies will be small (and hence close to experimental values) when $m_\pi$ is large.  Amongst the various crystal solutions, we find that the sheet, chain and $\alpha$-crystals remain close in energy, with stable order, but the gap to the $1/2$-crystal (which always has highest energy) grows  with $m_\pi$.
\begin{figure}[t]
    \centering
    \includegraphics[width=0.7\textwidth]{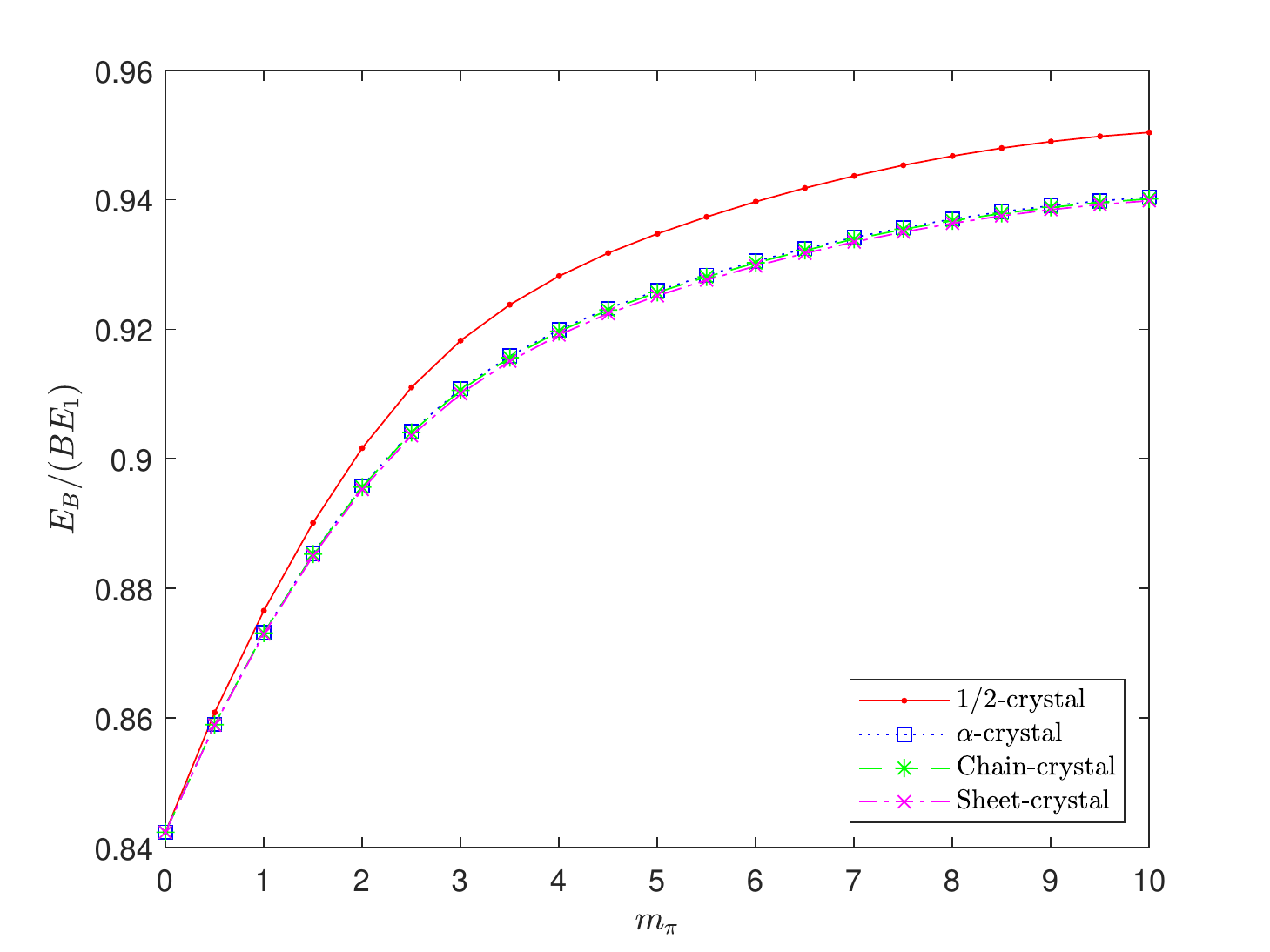}
    \caption{Comparison of the normalized energies per baryon per unit cell of the four Skyrme crystals for increasing pion mass $m_\pi$. Energies are presented in units of the energy of the $B=1$ skyrmion at the relevant pion mass (which grows monotonically with $m_\pi$).}
\label{fig: Massive Skyrme crystals - Comparison}
\end{figure}


\section{Skyrme crystals at prescribed average baryon density}
\label{sec: Fixed baryon density variations}

If we are to use Skyrme crystals as a model of dense nuclear matter (for example, in astrophysical contexts) it is important to understand the properties of the lowest energy configuration among all those with a fixed average baryon {\em density}, treating this density as a parameter of our system.
This problem was first approached by Hen and Karliner \cite{Hen_2008} in the context of the baby Skyrme model.
Therein they extremized the baby Skyrme energy functional with respect to variations of the period lattice at a constant skyrmion density.
This method was carried out at various densities, producing an energy-density curve.
However, they did not address the nature of the critical points they obtained, stating that they could in fact turn out to be maxima or saddle points.
Our method is similar but it is more general and robust.

Let us fix $B$, the baryon number per unit cell. Then the average baryon density of a configuration 
$(\phi,g)$ is 
\beq
\rho_B=\frac{B}{\int_{\T^3}\vol_g}=\frac{B}{\sqrt{\det g}}.
\eeq
Hence, finding the minimal crystal with fixed baryon density (and baryon number $B$ per unit cell) amounts to minimizing $E:\MM\ra\R$ over a level set of $\det g$. Once again, we can address the partial minimization problem where we fix the field $\phi:\T^3\ra SU(2)$ (assumed to be $C^1$ and somewhere immersive) and
a density $\rho_B=B/\nu$ then seek a minimum of $E(\phi,\cdot):\det^{-1}(\nu^2)\ra\R$. It turns out that, like the unconstrained minimization problem studied in section \ref{sec:eumm}, this problem has a unique global minimum and no other critical points.

\begin{prop} Let $\phi:\T^3\ra\SU(2)$ be a fixed $C^1$ map that is immersive somewhere and $\nu>0$ be a constant. Then the function $\SPD_3\supset\det^{-1}(\nu^2)\ra\R$ mapping each flat metric $g$  on $\T^3$ of volume $\nu$ to the Skyrme energy $E(\phi,g)$ attains a unique global minimum, and has no other critical points.
\end{prop}

\begin{proof} As before, it is equivalent to prove that the associated function
\beq
\wt{E}:{\det}^{-1}(\nu^{-1})\ra\R,\quad \wt{E}=E\circ\sigma^{-1}
\eeq
attains a unique global minimum and has no other critical points, where $\sigma:\SPD_3\ra\SPD_3$ is the diffeomorphism $\sigma(g)=g/\sqrt{\det g}$. 
Now
\beq
\wt{E}(\Sigma)=\Tr(H\Sigma^{-1})+\Tr(F\Sigma)+C_0\nu+C_6\nu^{-1}
\eeq
where $H,F\in\SPD_3$ and $C_0,C_6\in[0,\infty)$ are the $\phi$-dependent constants previously defined.  Existence of a global minimum of $\wt{E}$ follows {\it mutatis mutandis} from Proposition \ref{prop1}, since the bound \eqref{thebound} still holds irrespective of the extra constraint
$\lambda_1\lambda_2\lambda_3=\nu^{-1}$ (equivalent to $\det\Sigma=\nu^{-1}$). This confines the minimizing sequence to a compact subset of the hypersurface $\lambda_1\lambda_2\lambda_3=\nu^{-1}$ in $(0,\infty)^3\times O(3)$, whence a convergent subsequence can be extracted, whose limit attains the infimum of $\wt{E}$ by continuity.

It remains to prove uniqueness. Assume towards a contradiction that
$\wt{E}:{\det}^{-1}(\nu^{-1})\ra\R$ has two distinct critical points
$\Sigma_*$, $\Sigma_{**}$. Then there exists a geodesic
\beq
\gamma(t)=A\exp(\xi t)A^T
\eeq
in $(\SPD_3,G)$ with $\gamma(0)=\Sigma_*$ and $\gamma(1)=\Sigma_{**}$. Now
\beq
\det(\gamma(t))=(\det A)^2e^{t\Tr\xi}
\eeq
and $\det(\gamma(0))=\det(\gamma(1))$, so $\Tr\xi=0$ and we conclude that $\det(\gamma(t))$ is constant. Hence the geodesic $\gamma$ remains on the level set ${\det}^{-1}(\nu^{-1})$. Since $\Sigma_*$, $\Sigma_{**}$ are critical points of $\wt{E}:{\det}^{-1}(\nu^{-1})\ra\R$, and $\gamma$ is tangent to the level set for all $t$, $(\wt{E}\circ\gamma)'(0)=0=(\wt{E}\circ\gamma)'(1)$, so by Rolle's Theorem $(\wt{E}\circ\gamma)''$ vanishes somewhere on $(0,1)$, contradicting the convexity of $\wt{E}$ (Proposition \ref{prop2}). Hence no second critical point may exist.
\end{proof}

In the course of the proof above we established that all level sets of $\det$ are connected totally geodesic submanifolds of $(\SPD_3,G)$, and hence the restriction of $\wt{E}$ to any such level set is strictly convex. Note that, in general, the restriction of a convex function to a submanifold may fail to be convex, so total geodesicity of the level sets is crucial here. 

We can again solve the minimization problem for $E:{\det}^{-1}(\nu^2)\ra\R$ numerically by arrested Newton flow, but now we must take care to project the gradient of $E$ tangent to the level set. Given a curve $g(t)$ in $\det^{-1}(\nu^2)$,
\beq
\frac{d\: }{dt}\bigg|_{t=0}\det g(t)=\nu^2\Tr(g(0)^{-1}\dot{g}(0))=0
\eeq
so $\dot{g}(0)$ is orthogonal to $g(0)^{-1}$ with respect to the Euclidean metric $\ip{X,Y}=\Tr(X^TY)$. Hence $T_g\det^{-1}(\nu^2)=\ip{g^{-1}}^\perp$.
Now
\beq
E(g(t))=\nu\Tr(Hg^{-1})+\frac{\Tr(Fg)}{\nu}+C_0\nu
+\frac{C_6}{\nu},
\eeq
and hence
\beq
\d E_g(v)=\frac{1}{\nu}\ip{F-\nu^2 g^{-1}H g^{-1},v}.
\eeq
It follows that, with respect to the metric on $\det^{-1}(\nu^2)$ induced by the Euclidean metric,
\beq
(\grad E)(g)=\frac1\nu\left\{F-\nu^2 g^{-1}H g^{-1}
-\frac{\Tr((F-\nu^2 g^{-1}Hg^{-1})g^{-1})}{\Tr(g^{-1}g^{-1})}g^{-1}\right\}.
\eeq
We solve the Newton flow $\ddot g=-(\grad E)(g)$ numerically, projecting $g$ back onto $\det^{-1}(\nu^2)$ after each time step by radial dilation
($g\mapsto (\nu^2/\det g)^{1/3}g$), arresting if $E(g(t+\delta t))>E(g(t))$, and terminating if the sup norm of $\grad E$ falls below a prescribed tolerance. As for the unconstrained problem, we apply this algorithm after each iteration of the arrested Newton flow for the field $\phi:\T^3\ra S^3$.

\begin{figure}[t]{}
    \centering
    \includegraphics[width=\textwidth]{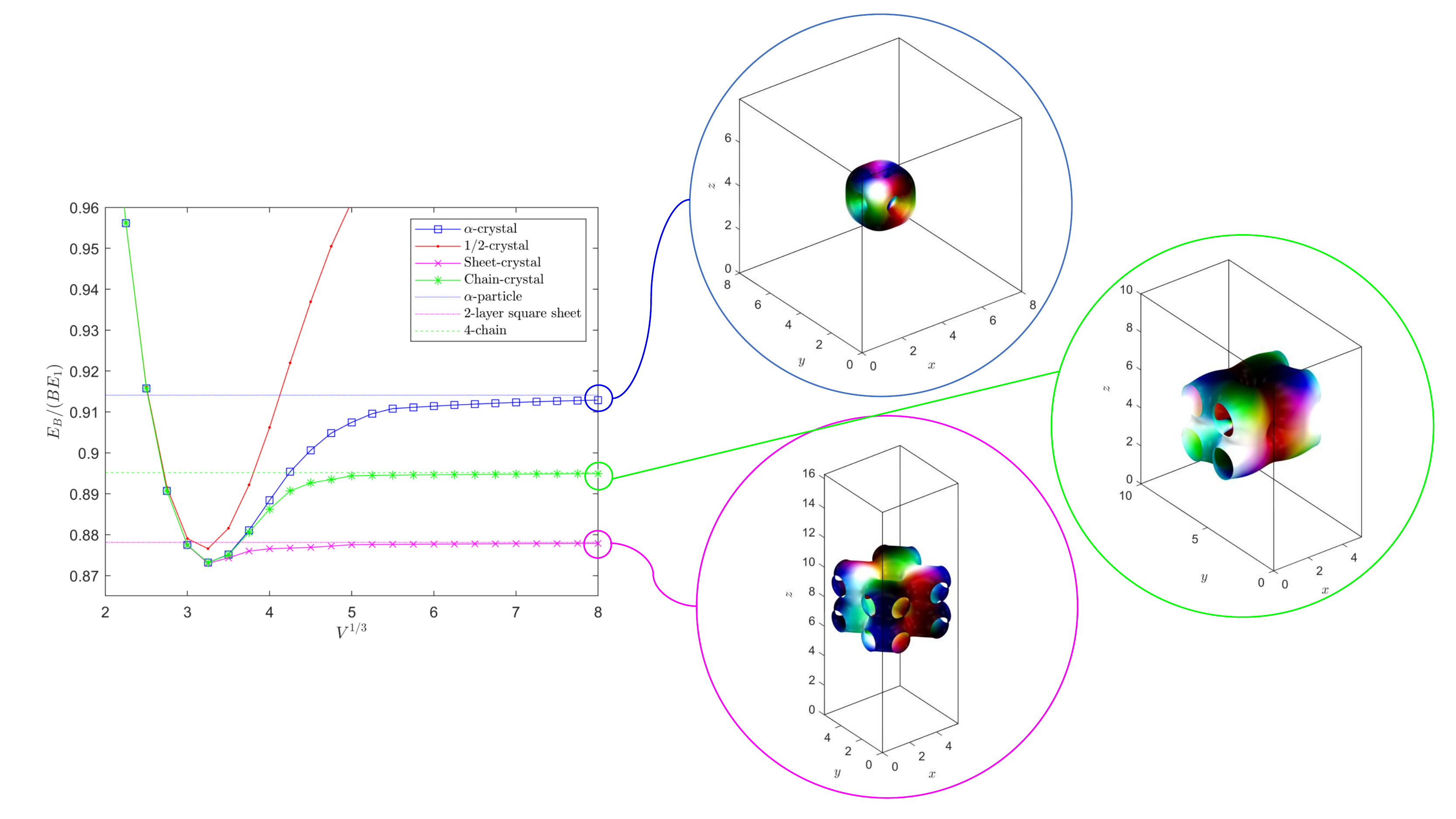}
    \caption{The energy per baryon per unit cell of the Skyrme crystals in the model with pion mass $m_\pi=1$ as a function of cell volume.}
    \label{fig: Skyrme crystals}
\end{figure}

Applying this approach at various densities to the four crystals found in section \ref{sec: Skyrme crystal solutions} at pion mass $m_\pi=1$, we observe that the three lower energy crystals tend to finite-energy solutions at low densities.
The $\alpha$-crystal tends to the $B=4$ $\alpha$-particle skyrmion on $\mathbb{R}^3$ \cite{Braaten_1990}, see figure \ref{fig: Skyrme crystals}.
This phase transition has already been observed by Silva Lobo \cite{SilvaLobo_2010} in the massless model and by Adam \textit{et al.}\ \cite{Adam_2022} in the massive model with sextic term.
The sheet-crystal tends to a double-layered square sheet on $\mathbb{T}^2 \times \mathbb{R}$, similar to the $2$-wall massless solution found by Silva Lobo and Ward \cite{SilvaLobo_2009}.
Finally, the chain-crystal becomes a linear chain on $\mathbb{R}^2 \times S^1$, which appears to be a previously unknown solution.

At low densities the sheet solution is clearly energetically preferred over other solutions.  This qualitative result was predicted earlier in \cite{Park_2019}.  However, there are some important differences between between our result and \cite{Park_2019}.  Our result comes from a minimisation over all Skyrme fields and lattice geometries, whereas \cite{Park_2019} used the more restrictive Atiyah-Manton approximation for the Skyrme field and assumed symmetric lattice geometries.  Second, our minimal-energy Skyrme sheet has a square geometry, whereas those constructed in \cite{Park_2019} had a hexagonal geometry.  Finally, our results are for the model with $m_\pi=1$, whereas \cite{Park_2019} considered $m_\pi=0$.

As one might expect, the four crystals become energetically indistinguishable in the large density limit. As far as we can determine, the curves in figure \ref{fig: Skyrme crystals} never cross, so the crystals maintain their energy ordering at all densities.

\section{Conclusion}
\label{sec:conc}

In this paper, we developed methods to obtain Skyrme crystals in a general class of Skyrme models, and 
presented a detailed numerical study of crystals in the standard Skyrme model with massive pions. To achieve this, we minimized the model's energy with respect to variations of both the field and its period lattice in $\R^3$.
A key idea is to reformulate the latter variation as a variation over all flat metrics on the fixed unit torus $\T^3$. We obtained strong results on the partial minimization problem in which the field is fixed and only the metric varied: under a mild nondegeneracy assumption on the field, there exists a unique flat metric that globally minimizes the Skyrme energy, and no other critical metrics. This result holds also if we constrain the problem to vary only over metrics of fixed volume, a variant relevant to constructing Skyrme crystals of prescribed average baryon density. Our methods impose no symmetry on the period lattice {\it a priori}, and hence go beyond previous studies which imposed a cubic unit cell. 

We find that the minimal energy crystal (with baryon number $4$ per unit cell) has trigonal but not cubic periodicity. At low densities it tends to a double sheet solution. The next lowest energy crystal is also trigonal and not cubic, tending to a chain solution at low densities. Both these crystals are new. Above them in energy are two already known solutions, the $\alpha$-crystal and the $1/2$-crystal. All these crystals, except the most energetic, the $1/2$-crystal, have anisotropic isospin inertia tensors. The existence of four distinct crystals can be understood semi-analytically by means of the Principle of Symmetric Criticality and the Inverse Function Theorem.

The methods detailed in this paper could be applied to the study of isospin asymmetric nuclear matter within the Skyrme model.
The next step would be to investigate neutron crystals by considering the quantum corrections to the energy due to the quantization of the isospin degrees of freedom, and improve on the work done on the massless model by Baskerville \cite{Baskerville_1996}.
In particular, one could determine ``nuclear pasta'' phases in neutron stars  \cite{Ravenhall_1983} by considering the quantization of generalized Skyrme crystals in the low density regime.
The chain-crystal we have found could correspond to the so-called ``spaghetti'' phase, and the sheet-crystal the ``nuclear lasagne''.


\section*{Acknowledgments}

We would like to thank the \href{http://solitonsatwork.net/?}{solitons@work} community for discussions following the Solitons (non)Integrability Geometry X \href{http://th.if.uj.edu.pl/~wereszcz/sig10.html}{(SIG X)} and Geometric Models of Nuclear Matter \href{https://www.kent.ac.uk/smsas/personal/skyrmions/GMNM2022/index.html}{(GMNM)} conferences in June/July 2022, which motivated part of this paper.
P.N.L. is supported by a PhD studentship from UKRI, Grant No. EP/V520081/1.


\appendix



\bibliographystyle{JHEP.bst}

%
\bibliography{crystals.bib}

\end{document}